\documentclass[a4paper,11pt]{article}

\usepackage[sort&compress,square,comma,numbers]{natbib}
\usepackage{amsfonts,amssymb,amsthm}
\usepackage[all]{xy} % draw graphs
\usepackage{mathrsfs} 
\usepackage{url}
\usepackage{color,xcolor}
\usepackage{enumerate}
\usepackage{blkarray}
\usepackage{subfig}
\usepackage[all]{xy}
\usepackage{graphicx}
\usepackage{hyperref}
\usepackage[version=3]{mhchem}
\usepackage{tikz}
\usetikzlibrary{decorations}
\usetikzlibrary{shapes}
\usepackage{relsize}

\newcommand{\mmS}{\mathcal{S}}
 \newcommand{\mmC}{\mathcal{C}}
\newcommand{\mmN}{\mathcal{N}}
\newcommand{\mmE}{\mathcal{E}}
\newcommand{\mmEu}{\mathcal{E}_{\mathcal{U}}}
\newcommand{\mmR}{\mathcal{R}}
\newcommand{\mmU}{\mathcal{U}}
\newcommand{\mmH}{\mathcal{H}}
\newcommand{\mmRu}{\mathcal{R}_{\mathcal{U}}}
\newcommand{\mmG}{\mathcal{G}}
\newcommand{\mmGu}{\mathcal{G}_{\mmU}}

\renewcommand{\k}{\kappa}
\newcommand{\V}{\mathcal{B}}

\newcommand{\R}{{\mathbb R}}
\newcommand{\Z}{{\mathbb Z}}

\newcommand{\Rnn}{\mathbb{R}_{\geq 0}}
\newcommand{\Rp}{\mathbb{R}_{> 0}}

\newcommand{\st}{\,\mid \,}
 
\DeclareMathOperator{\supp}{supp}

\newtheorem{lemma}{Lemma}

\newtheorem{theorem}{Theorem}

\theoremstyle{definition}
\newtheorem{definition}{Definition}
\newtheorem{example}{Example}

\begin{document}
 \title{Nonnegative linear elimination for chemical reaction networks}

\author{Meritxell S\'aez$^{1}$, Carsten Wiuf$^{1}$, Elisenda Feliu$^{1,2}$}
\date{\today}

\footnotetext[1]{Department of Mathematical Sciences, University of Copenhagen, Universitetsparken 5, 2100 Copenhagen, Denmark.}
\footnotetext[2]{Corresponding author: efeliu@math.ku.dk}

\maketitle

\begin{abstract}
We consider linear elimination of variables in steady state equations of a chemical reaction network. Particular subsets of variables corresponding to sets of so-called reactant-noninteracting species, are introduced. The steady state equations for the variables in such a set,  taken together with potential linear conservation laws in the variables, define a linear system of equations.
 We give conditions that guarantee that the solution to this system is nonnegative, provided it is unique.
 The results are framed in terms of spanning forests of a particular multidigraph derived from the reaction network and thereby conditions for uniqueness and nonnegativity of a solution are derived by means of the multidigraph. Though our motivation comes from applications in systems biology, the results have general applicability in applied sciences.
\end{abstract}

\section{Introduction}
Systems of (parameterized) Ordinary Differential Equations (ODEs) are commonly used to describe complex dynamical systems of interacting species in cellular and systems biology, epidemiology and ecology. As a first description of the dynamics of a system, it is of interest to have a characterization of the steady states, their number and properties, for varying parameter values and constants of the system. This is rarely straightforward to obtain, except in simple cases.

Still, in some cases, a full or partial parameterization of the  manifold of the nonnegative steady states can be obtained. This manifold is typically of positive dimension because of conserved linear quantities among the variables of the system (concentrations or abundances of species). Hence,  the nonnegative steady states might be considered as the points in the intersection of the manifold with the equations for the conserved quantities (defined by the initial value of the system).
A partial or full parameterization of the manifold might further reduce significantly the number of free variables and provide crucial information about its structure. 
  For example, in \cite{feliu:intermediates},  classes of systems sharing a common core are studied and a partial parameterization is obtained in terms of the nonnegative steady states of the core system. In other cases, characterization of multistationarity  \cite{Conradi:PLOSCOMP} and of stability properties of steady states \cite{Feinbergss,feinberg-def0,feinberg-horn-open,hornjackson} are obtained by means of a full parameterization. The present paper is concerned about the existence and computation of partial (full) parameterizations.

Most work in this area has been done in the particular case of systems of chemical reactions --so-called reaction networks-- with mass-action kinetics. In this context, the steady states are the nonnegative solutions to a system of polynomial equations with parameterized coefficients. Standard computational algebra tools, such as Gr\"obner bases, might be applied to determine the steady states, but it  is rarely straightforward to assess nonnegativity and  it does  not seem to be a viable approach in general. For special classes of reaction networks full parameterizations have been obtained. In \cite{hornjackson,Craciun-Sturmfels}, a  parameterization of the manifold of the positive steady states of a complex balanced reaction network is provided.  A generalization of this result can be found in \cite{PerezMillan} for systems with toric steady states and in \cite{muller} for complex balanced steady states with generalized mass-action kinetics.

Species that never appear together at the same side of a reaction are introduced  in  \cite{Fel_elim} as noninteracting species.  The variables corresponding to the noninteracting species can be expressed in terms of the remaining variables and   this  provides in general a partial parameterization of the steady state manifold.
Post-translational modification systems form an important class of models that falls into this framework  \cite{TG-rational,fwptm}.

The key idea in the elimination of the noninteracting species is that their associated steady state equations define a linear  system with a unique solution. The solution might  be expressed as a rational function with positive coefficients in the remaining variables and parameters of the system. For linearity of the system, it is only required that  the noninteracting species  do not interact in the  reactant (left-hand side) of any reaction in the network. Perhaps surprisingly, this is not sufficient to assert that the solution is also nonnegative. For that, further conditions are required.

Here we study the existence of partial parameterizations of the manifold of the nonnegative steady states in terms of reactant-noninteracting species, that is, species that never appear together in the reactant of a reaction, but potentially do it in the product (right-hand side) of the reaction.
Our results build on  previous work on nonnegative solutions to systems of linear equations   \cite{Saez:PosSol} and contain the case of noninteracting species as a special case. We show by example that the partial (full) parameterizations can be obtained in even large systems.  An appealing feature of the work is that it is essentially graphical in nature. This makes it very easy to apply in concrete examples and decide on nonnegativity of the solution. For moderately sized systems, the relevant graphs can be drawn and analyzed by hand.

The paper is organised in the following way. In section \ref{sec:reactnet},  basic concepts of reaction network theory are introduced, and in section \ref{sec:eliminationsystem},  reactant-noninteracting sets of species and the so-called elimination system are defined.  Section \ref{sec:muldi} contains the explicit expressions of the solution in terms of the labels of a multidigraph, and  the main statements about uniqueness and  non-negativity of solutions. Two additional graphs that can be used to find reactant-noninteracting subsets and to study the nonnegativity of the solution are also discussed. In section \ref{sec:examples}, we give examples based on models of real biological systems. Finally,  section \ref{sec:proofs} contains proofs of the main results.

\section{Reaction networks}\label{sec:reactnet}
Let $\Rnn$ and $\Rp$ denote  the sets of nonnegative and positive real numbers, respectively, and define  $\Z_{\geq 0}$ analogously. A vector $x\in\R^n$ is \emph{positive} (resp.~\emph{nonnegative}) if $x_i>0$ (resp.~$x_i\ge 0$) for all $i=1,\ldots,n$.
For $x,y\in \R^n$, $x\cdot y$ denotes the scalar product associated with the Euclidean norm, and  $\langle v_1,\dots,v_r\rangle$ denotes the vector subspace generated by  $v_1,\dots,v_r\in \R^n$. 
 The power set of a finite set $W$ is denoted by $\mathcal{P}(W)$.

A \emph{reaction network} on an ordered  finite set $\mmS=\{S_1,\dots,S_n\}$ of \emph{species} is a digraph $(\mmC, \mmR)$ such that $\mmC\subseteq \Z_{\geq 0}^{\mmS}$. The nodes   are  called \emph{complexes}, and the edges   \emph{reactions}.  Furthermore, the source (resp.~target) of a reaction $r\in\mmR$ is called the reactant (resp.~product) of the reaction and is denoted by $y_r$ (resp.~$y'_r$). Since $\Z_{\ge 0}^\mmS\subseteq \R^{\mmS} \cong \R^n$, the complexes might be considered vectors of $\R^n$.  

The \emph{stoichiometric coefficient}  of the species $S_i\in\mmS$ in $\eta\in\mmC$ is  $\eta_i$, $i=1,\ldots,n$.  A species $S_i$ is in $\eta\in\mmC$ if $\eta_i>0$, and $S_i$ is in a reaction if it is in the reactant or the product of the reaction. For convenience, if a reaction network is given by specifying its reactions, we implicitly take the set of complexes to consists of all reactants and products and the set of species to consist of the species in the reactions. 
A pair of species $S_i,S_j$ \emph{interact}  if $S_i$ and $S_j$ are in the same complex. 
If $i=j$, this is understood to mean that $S_i$ appears with stoichiometric coefficient at least $2$ in a complex, and we say that $S_i$ self-interacts.
 For  $\mmS'\subseteq \mmS$, we say that $S\in \mmS'$ \emph{ultimately produces} $S'\in\mmS'$ \emph{via} $\mmS'$ if there exist distinct $S_{i_1},\dots, S_{i_\ell}\subseteq \mmS'$ with $S_{i_1}=S$, $S_{i_\ell}=S'$ and for $j=1,\dots,\ell-1$, $S_{i_j}$ is in the reactant and $S_{i_{j+1}}$ is in the product of some reaction.

To model the dynamics of a reaction network $(\mmC,\mmR)$ on $\mmS$ we introduce an edge labeling of $\mmR$, called a \emph{kinetics}. The label of  $r\in\mmR$ is a function $\kappa_r\colon \Omega_0 \rightarrow \R_{\geq 0}$,  called the \emph{rate function} of $r$, where $\Omega_0\subseteq \R^n_{\geq 0}$ and $\kappa_r(\Omega_0\cap\R^n_{>0})\subseteq \R_{>0}$. For convenience, we order the set of reactions $\mmR=\{r_1,\ldots,r_p\}$ and  let $\k=(\k_{r_1},\ldots,\k_{r_p})$ denote a kinetics. Hence $\k$ is a function from $\Omega_0$ to $\Rnn^p$. 
We often write $\k_i$ instead of $\k_{r_i}$. 

Under \emph{mass-action kinetics}, the rate functions are 
\[
\k_r\colon \Rnn^n\to\Rnn,\quad \kappa_r(x)=k_rx^{y_r}=k_r\prod\limits_{i=1}^nx_i^{(y_r)_i},\qquad r\in\mmR,
\]
where  $k_r>0$  is  the \emph{reaction rate constant} of reaction $r$. We use $k_{r_i}=k_i$ as the edge labeling of $\mmR$ in this case. By convention, $0^0=1$.

Given a reaction network $(\mmC,\mmR)$ on $\mmS$ with kinetics $\k$,  the ODE system modeling the evolution of the species concentrations over time is given as
\begin{equation}\label{eq:ODE}
\dot{x}=\sum_{r\in\mmR}\k_r(x)(y'_r-y_r), \qquad x\in \Omega_0,
\end{equation}
where $\dot{x}=(\dot{x}_1,\dots,\dot{x}_n)$ is the derivative of $x=(x_1,\ldots,x_n)$ with respect to time, and $x_i$ is  the concentration of species $S_i\in\mmS$. Explicit reference to time is omitted. The \emph{steady states} of the ODE system \eqref{eq:ODE} are the solutions to the system of equations 
\begin{equation}\label{eq:sseq}
\sum_{r\in\mmR}\k_r(x)(y'_r-y_r)=0, \qquad x\in\Omega_0, 
\end{equation}
referred to as the \emph{steady state equations}.

The \emph{stoichiometric subspace} of a reaction network $(\mmC,\mmR)$ on $\mmS$ is the  vector subspace of $\mathbb{R}^n$  given by
\[
S=\big\langle y'_r-y_r\st r\in \mathcal{R} \big\rangle\subseteq \R^n.
\]
It follows from \eqref{eq:ODE} that $\omega\cdot \dot{x}=0$ for any $\omega\in S^{\bot}$. Thus the quantity  $T=\omega\cdot x$, named the \emph{conservation law} of $\omega\in S^{\bot}$ with total amount $T$, is conserved over time.  Further, for every $x_0\in \R^n_{\geq 0}$ the trajectory belongs to the invariant linear variety $x_0+S$, defined by the equations $\omega\cdot x = \omega \cdot x_0$ for all $\omega\in S^\perp$, and it remains nonnegative. The polyhedra $(x_0+S)\cap \R^n_{\geq 0}$ are called \emph{stoichiometric compatibility classes}. Because of this invariance, it has been of interest to study the steady states within each stoichiometric compatibility class.

\begin{example}
Consider the following  reaction network with kinetics $\k$:
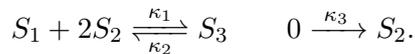
\[
S_1+2S_2  \ce{<=>[\k_1][\k_2]}S_3 \qquad  0\ce{->[\k_3]}S_2.
\]
Here, $\mmS=\{S_1,S_2,S_3\}$ and  $\mmC=\{S_1+2S_2,\, S_3,\, 0,\, S_2\}$. The  ODE system  is
\begin{align*}
\dot{x}_1=&-\k_1(x)+\k_2(x), &
\dot{x}_2=&-2\,\k_1(x)+2\,\k_2(x)+\k_3(x),&
\dot{x}_3=& \k_1(x)-\k_2(x).&
\end{align*}
Since $\dim(S^{\bot})=1$, there is one independent conservation law, for example $T_1=x_1+x_3$.
\end{example}

\section{Reactant-noninteracting sets and linear kinetics}\label{sec:eliminationsystem}
Let a  reaction network $(\mmC,\mmR)$ on $\mmS$ be given, and let $\mmU\subseteq\mmS$. 
For convenience, we assume $\mmU=\{U_1,\dots ,U_m\}$, $\mmS\setminus \mmU=\{X_1,\dots,X_{m'}\}$  with $m'=n-m$, and order $\mmS$ such that 
\[
\mmS=\{ U_1,\dots,U_m,X_1,\dots,X_{m'}\}.
\] 
Similarly, the vector of concentrations is given as $(u,x)$ with $u=(u_1,\ldots,u_m)$ and $x=(x_1,\ldots,x_{m'})$.
Let $\rho\colon\R^n\rightarrow \R^m$ be the projection onto the first $m$ components, mapping $(u,x)$ to $u$.
Define the \emph{support} of  a vector $\omega\in \R^n$ as $\supp(\omega)=\{S_i\st \omega_{i}\neq 0\}\subseteq \mmS$, and let
\begin{equation*} 
S_\mmU^\perp=\{\omega\in S^\perp\st \supp(\omega)\subseteq \mmU  \} \subseteq\   S^\perp.
\end{equation*}
(Note that $S_\mmU^\perp$ is not necessarily the orthogonal complement of the vector subspace generated by the vectors in $S$ with support in $\mmU$.) For  $(u_0,x_0)\in \R^n_{\geq 0}$,  the linear variety $(u_0,x_0)+S$ is a subvariety of 
  $(u_0,x_0)+(S_\mmU^\perp)^\perp$. Given any basis $\omega^1,\dots,\omega^d$ of $S_\mmU^\perp$,  $(u_0,x_0)+(S_\mmU^\perp)^\perp$ is defined by the equations
\begin{equation*}  
\rho(\omega^i) \cdot u = \rho(\omega^i) \cdot u_0, \quad i=1,\dots,d, \qquad (u,x)\in \R^n,
\end{equation*}
hence $(u_0,x_0)+(S_\mmU^\perp)^\perp$ is independent of $x_0$.

The focus of this work is to study, for a given $x\in \R^{m'}_{\geq 0}$ and $u_0\in \R^m_{\geq 0}$, the solutions to the following system of equations
\begin{equation}\label{eq:elimsyst}
\left\{\begin{array}{cl}
\dot{u}_i=0 &\text{ for }\quad i=1,\dots,m, \\[2pt]
\rho(\omega^j)\cdot u= \rho(\omega^j) \cdot u_0 & \text{ for }\quad j=1,\dots,d.
\end{array}\right.
\end{equation}
Since every vector $\omega\in S_\mmU^\perp$ defines a linear relation among the equations $\dot{u}_i=0$, 
$d$ of these equations are redundant. Removal of redundant equations leads to a system with $m$ equations and $m$ variables, whose solution set is independent of the choice of a basis of $S_\mmU^\perp$. 
 The solutions to this system might be used to find a parameterisation of the steady state variety and provide a first  step in finding the steady states in a particular stoichiometric compatibility class.
In order to study solutions to \eqref{eq:elimsyst}, we focus on certain classes of sets $\mmU$ and impose restrictions on the kinetics.

\begin{definition}\label{def:nonint}
A set $\mmU\subseteq \mmS$ is \emph{noninteracting} if it does not contain a pair of interacting species nor self-interacting species.
A set $\mmU\subseteq \mmS$ is \emph{reactant-noninteracting} if it does not contain a pair of species interacting nor self-interacting species in the \emph{reactant} of any reaction.
\end{definition}
 
We assume the domain  of a kinetics is of the form $\Omega_0=\R_{\ge 0}^m\times \Omega$ with $\Omega\subseteq \R^{m'}_{\geq 0}$.

\begin{definition}\label{defUlinear} Let $\mmU$ be a reactant-noninteracting set, and $\kappa\colon \R^m_{\geq 0}\times \Omega\rightarrow \R^{p}_{\geq 0}$ be a kinetics with $\Omega\subseteq \R^{m'}_{\geq 0}$. The kinetics $\kappa$ is \emph{$\mmU$-linear} if, for every $r\in \mmR$, there exists a  function $v_r\colon\Omega\rightarrow \R_{\geq 0}$ such that $v_r(\Omega\cap \Rp^{m'})\subseteq \Rp$, and 
\[
\kappa_r(u,x)=\left\{\begin{array}{ll} 
u_iv_r(x) & \text{ if }\rho(y_r)_{i}=1, \\ 
\phantom{u_i}v_r(x)& \text{ if }\rho(y_r)= 0.
\end{array}\right.
\]
\end{definition}

For a $\mmU$-linear kinetics and fixed $x\in \Omega$, the system $\dot{u}=0$ is of the form
\[ \widetilde{A}(x) u + \widetilde{b}(x)=0,\]
with $\widetilde{A}(x)=(\widetilde{a}_{ij}(x))_{i,j\in\{1,\ldots,m\}}$, $\widetilde{b}(x)=(\widetilde{b}_i(x))_{i\in\{1,\ldots,m\}},$ given by
\begin{equation}\label{eq:sseqU}
\widetilde{a}_{ij}(x)=\sum\limits_{r\in\mmR,\,  (y_r)_j\neq 0}v_r(x)(y'_r-y_r)_i, \qquad 
\widetilde{b}_i(x)=\sum\limits_{r\in\mmR,\,  \rho(y_r)=0}v_r(x)(y'_r)_i.
\end{equation}
By letting $T_{u_0}=(\rho(\omega^1)\cdot u_0,\dots,\rho(\omega^d)\cdot u_0)$, and after removal of $d$ redundant equations among $\dot{u}=0$, system \eqref{eq:elimsyst} is  linear with   $m$ equations in the $m$ variables $u_1,\dots,u_m$,
\begin{equation}\label{eq:linelimsyst}
A(x) u+b(x,T_{u_0})=0.
\end{equation}
We refer to this system as the \emph{elimination system} (associated with $\mmU$ and $u_0$), and note that it is defined up to a choice of  basis and the removal of redundant equations.
If $\det(A(x))\not=0$, then \eqref{eq:linelimsyst} has  a unique solution. Our aim is to find the solution and decide whether it  is nonnegative.

\refstepcounter{theorem}\label{ref:reactnet-example}
\newcounter{mycounterReactnet}
\renewcommand{\themycounterReactnet}{\getrefnumber{ref:reactnet-example}\,(part\,\Alph{mycounterReactnet})}
\newtheorem{myexampleReactnet}[mycounterReactnet]{Example}

\begin{myexampleReactnet}\label{example:reactnet}
Consider the  reaction network with  kinetics $\k$,
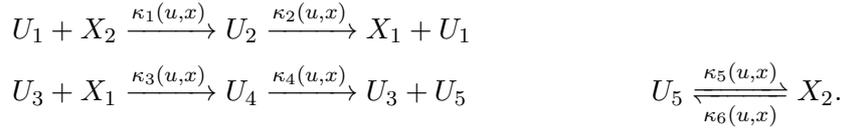
\begin{align*}
U_1+X_2\ce{->[\k_1(u,x)]}U_2\ce{->[\k_2(u,x)]}&X_1+U_1&\\
U_3+X_1\ce{->[\k_3(u,x)]}U_4\ce{->[\k_4(u,x)]}&U_3+U_5& U_5\ce{<=>[\k_5(u,x)][\k_6(u,x)]}X_2.
\end{align*}
The set $\mmU=\{U_1,\dots,U_5\}$ is reactant-noninteracting, but not noninteracting because $U_3$ and $U_5$ are both in the product of one reaction. If $\k$ is $\mmU$-linear, then 
$\k_i(u,x)= u_i v_i(x)$ for $i=1,\dots,5$ and $\k_6(u,x)= v_6(x).$

A basis of $S_\mmU^\perp$ is composed by the nonnegative vectors $\omega^1=(1,1,0,0,0,0,0)$ and $\omega^2=(0,0,1,1,0,0,0)$. Given $u_0\in \R^5_{\geq 0}$, we have $T_{u_0} = (T_1,T_2)= (u_{0,1}+u_{0,2},u_{0,3}+u_{0,4})$.
After removing $\dot{u}_2=0,\dot{u}_4=0$, the elimination system \eqref{eq:linelimsyst} becomes:
\begin{align}
u_1+u_2  -T_1 &= 0  & (\rho(\omega^1)\cdot u = T_1)\  \nonumber\\ 
-v_1(x) u_1 + v_2(x) u_2 & =0 & (\dot{u}_1=0)\ \nonumber \\
u_3 + u_4 - T_2 &= 0 &  (\rho(\omega^2)\cdot u = T_2)\   \label{eq:elimsyst_example}\\
-v_3(x) u_3 + v_4(x) u_4 &=0 & (\dot{u}_3=0)\ \nonumber \\ 
v_4(x)u_4 -v_5(x) u_5 + v_6(x) &=0 & (\dot{u}_5=0).\nonumber
\end{align}
Therefore, we have
\[ A(x)= \begin{pmatrix}
1 & 1 & 0 & 0 & 0 \\
-v_1(x) & v_2(x) & 0 & 0 & 0 \\  
0 & 0 & 1 & 1 & 0 \\
0 & 0 & -v_3(x) & v_4(x) & 0 \\
0 & 0 & 0 & v_4(x) & -v_5(x) 
\end{pmatrix}, \qquad b(x,T_{u_0})= \begin{pmatrix}
-T_1 \\ 0 \\ -T_2 \\ 0 \\ v_6(x)
\end{pmatrix}.
\]

\end{myexampleReactnet}

\section{The multidigraph $\mmG_\mmU$}\label{sec:muldi}
In preparation for the main results, we introduce a  multidigraph with $m+1$ nodes  (a digraph where self-edges and parallel edges are allowed \cite{Saez:PosSol}).  A key point is that the first $m$ rows of the  Laplacian of this multidigraph agree with the  matrix $\widetilde{A}(x)$ extended by the vector $\widetilde{b}(x)$ in \eqref{eq:sseqU}. See Section \ref{sec:proofs} for details.
We assume a reaction network $(\mmC,\mmR)$ on $\mmS$ is given together with a reactant-noninteracting set $\mmU\subseteq\mmS$.

We define $\mmRu$ to be the set of reactions that involve species in $\mmU$,
\[ \mmRu=\{r\in \mmR \st   \rho(y_r)\neq 0 \text{ or }  \rho(y'_r)\neq 0\},\]
and $\Lambda_\mmU$ to be the subset of $\mmRu$ of the reactions such that 
 the product has at least two species in $\mmU$ or one self-interacting species in $\mmU$:
\[
\Lambda_\mmU=
  \left\{r\in\mmR \ \left|\  \sum\limits_{U_i\in\mmU}(y_r)_i= 1 \text{ and }\sum\limits_{U_i\in\mmU}(y'_r)_i> 1\right.\right\}.
\] 
If  $\mmU$ is noninteracting, then   $\Lambda_\mmU=\emptyset$.

\begin{definition}\label{def:graphU} 
Let $\mmU\subseteq \mmS$ be a reactant-noninteracting set, $\kappa$ a $\mmU$-linear kinetics and $x\in \Omega$. 
Let $\mmG_{\mmU}=(\mmN_{\mmU}, \mmE_{\mmU})$ be the labeled multidigraph with
$\mmN_{\mmU}=\mmU\cup \{*\}$ and $\mmEu=\mmEu^+\cup \mmEu^-$, where 
\begin{align*}
\mmEu^+= & \ \{ U_j \ce{->[(y_r')_iv_r(x)]} U_i \st r\in \mmRu,\ (y_r)_{j}=1\text{ and }  (y'_r)_{i}\neq 0\textrm{ for }i\neq j\}\ \cup \\
	& \ \{U_j \ce{->[v_r(x)]}*\st r\in \mmRu\text{ with }(y_r)_{j}=1\text{ and }  \rho(y'_r)= 0\}\ \cup \\
	& \ \{* \ce{->[(y'_r)_iv_r(x)]}U_i\st r\in \mmRu\text{ with }\rho(y_r)= 0\text{ and } (y'_r)_{i}\neq 0 \}\quad \text{ and}\\
\mmEu^-=	&\ \{U_j  \ce{->[-\lambda_rv_r(x)]} *  \st r\in \Lambda_\mmU\text{ with }(y_r)_{j}=1\text{ and } \lambda_r=\sum_{i=1}^m(y'_r)_{i}-1\}.
\end{align*}
\end{definition}

Explicit reference to $x$ in  $\mmGu$ is omitted.
Edges in $\mmEu^+$ have nonnegative labels and edges in $\mmEu^-$ have nonpositive labels. A label of an edge is zero only if $v_r(x)=0$, which happens only if $x$ has  zero entries. 
The multidigraph $\mmG_{\mmU}$ might have  parallel edges between any pair of nodes but no self-edges. 
An edge in $\mmEu$ corresponds to a unique reaction in $\mmR_\mmU$.  
 A reaction  $r\in\mmRu\setminus \Lambda_\mmU$ with  $\rho(y_r)\neq 0$ gives rise to one edge in $\mmEu^+$, while if $\rho(y_r)=0$, then $r$  corresponds to as many edges in $\mmEu^+$ as  there are species  in the product of $r$. 
Every reaction in $\Lambda_\mmU$ gives rise to one edge in $\mmEu^-$, and also one edge in the first subset of  $\mmEu^+$ for every species in $\mmU$ that is  in the product but not in the reactant of the reaction.
Hence if $\mmU$ is noninteracting, then $\mmE_\mmU^-$ is empty and all labels are of the form $v_r(x)$ since $(y'_r)_i$ is either $0$ or $1$.
In this case, the multidigraph $\mmGu$ coincides with the multidigraph  defined in \cite{Saez:reduction} after removal of self-edges.

\begin{myexampleReactnet}\label{example:reactnet0}
Here $\mmR=\mmR_\mmU$ and the set $\Lambda_\mmU$ consists of one reaction, namely $U_4\rightarrow U_3+U_5$. The multidigraph $\mmGu$ is
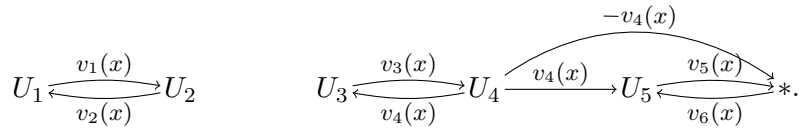
\begin{center}
\begin{tikzpicture}[inner sep=1.2pt]
\node (U3) at (0,2) {$U_3$};
\node (U4) at (2,2) {$U_4$};
\node (U5) at (4,2) {$U_5$};
\node (U2) at (-2,2) {$U_2$};
\node (U1) at (-4,2) {$U_1$};
\node (*) at (6,2) {$*$.};

\draw[->] (U3) to[out=10,in=170] node[above,sloped]{\footnotesize $v_3(x)$}(U4);
\draw[->] (U4) to[out=190,in=-10] node[below,sloped]{\footnotesize $v_4(x)$}(U3);
\draw[->] (U4) to[out=0,in=180] node[above,sloped]{\footnotesize $v_4(x)$}(U5);
\draw[->] (U4) to[out=35,in=145] node[above,sloped]{\footnotesize $-v_4(x)$}(*);
\draw[->] (U5) to[out=10,in=170] node[above,sloped]{\footnotesize $v_5(x)$}(*);
\draw[->] (*) to[out=190,in=-10] node[below,sloped]{\footnotesize $v_6(x)$}(U5);
\draw[->] (U1) to[out=10,in=170] node[above,sloped]{\footnotesize $v_1(x)$}(U2);
\draw[->] (U2) to[out=190,in=-10] node[below,sloped]{\footnotesize $v_2(x)$}(U1);
\end{tikzpicture}
\end{center}
\end{myexampleReactnet}

\bigskip
As will be stated below, the solution to the elimination system \eqref{eq:linelimsyst} can be expressed as a rational function in the labels of $\mmG_\mmU$, $T_{u_0}$ and $\omega^1,\dots,\omega^d$. In order to give these functions explicitly, some general definitions are required.
Consider a multidigraph $\mmG=(\mmN,\mmE)$ with no self-loops. 
Given $\mmN_0\subseteq \mmN$, $\mmG|_{\mmN_0}$ is the submultidigraph of $\mmG$ \emph{induced} by $\mmN_0$, that is, the multidigraph with node set $\mmN_0$ and all edges of $\mmG$ between pairs of nodes in $\mmN_0$.
A \emph{cycle} is a closed directed path with no repeated nodes.
A \emph{tree} is a directed subgraph of $\mmG$ such that the underlying undirected graph is connected and acyclic. A tree $\tau$ is \emph{rooted} at the node $N$, if $N$ is the only node without outgoing edges. In that case, there is a unique directed path from every node in $\tau$ to $N$. A \emph{forest} $\zeta$ is a directed subgraph of $\mmG$ whose connected components are trees. A tree (resp.~forest) is called a \emph{spanning tree} (resp.~\emph{spanning forest}) if the node set is $\mmN$.  For a spanning tree $\tau$ (resp.~a spanning forest $\zeta$) we use $\tau$ (resp.~$\zeta$) to refer to the edge set of the graph and to the graph itself indistinctly, as the node set in this case is $\mmN$.

If $\pi\colon\mmE \rightarrow \R$ is an edge labeling of $\mmG$, then any submultidigraph $\mathcal{G}'$ of $\mmG$ inherits a labeling from $\mmG$.  A labeling can be extended to  $\mathcal{P}(\mmE)$ by
\[
\pi\colon\mathcal{P}(\mmE)\to \R,\quad \pi(\mmE')=\prod_{e\in\mmE'}\pi(e) \quad \text{for}\quad\mmE'\subseteq \mmE.
\]
For a node $N$ of $\mmG$, $\Theta_\mmG(N)$ is the set of spanning trees of $\mmG$ rooted at $N$ and we let
\begin{equation*}
 \Upsilon_\mmG(N)=\sum_{\tau\in\Theta_\mmG(N)}\pi(\tau).
\vspace{-.3cm}
\end{equation*}
For $N_1,N_2,N_3\in \mmN$, define
\[
\Theta^{N_3}_\mmG(N_1,N_2)=\left\{\zeta\ \left| \begin{array}{l}\zeta \text{ is a spanning forest of $\mmG$ with two connected components:} \\
\text{a tree rooted at }N_2\text{ containing }N_1\text{ and a tree rooted at }N_3 \end{array}\right. \hspace{-0.2cm}\right\}
\]
 and
\begin{equation*}\label{eq:defUpsilon*}
\Upsilon^{N_3}_\mmG(N_1,N_2)=\sum_{\zeta\in\Theta^{N_3}_\mmG(N_1,N_2)} \pi(\zeta).
\end{equation*}

We return now to the multidigraph $\mmG_\mmU$ and the vector subspace $S_\mmU^\perp$.
The first result is a lemma that helps to understand the structure of $\mmGu$ imposed by the basis.

\begin{lemma}\label{lemma:cons_laws}
Let  $\omega\in S^\perp$ be nonnegative with support $\mmH\subseteq\mmS$. 
A reaction $r$ has a species in $\mmH$ in the product if and only if it has one in the reactant.
\end{lemma}
\begin{proof}
From  $0=\omega \cdot (y'_r-y_r)$, we have $\omega \cdot y'_r=\omega\cdot y_r \geq 0$, since $\omega$ is nonnegative. Using also that $y_r,y_r'$  are nonnegative vectors, $r$ has a species in $\mmH$ in the product  if and only if $\omega\cdot y'_r\neq 0$, if and only if $\omega\cdot y_r \neq 0$, if and only if $r$ has a species in $\mmH$ in the reactant.
\end{proof}

Any subset $\mmH\subseteq \mmU$ of a reactant-noninteracting (noninteracting) set is itself a reactant-noninteracting (noninteracting) set and a $\mmU$-linear kinetics is in particular $\mmH$-linear. Hence, the multidigraph $\mmG_\mmH$ is well defined.
Consider the  subsets of $\mmU$ defined by the supports of the basis vectors  $\omega^1,\dots,\omega^d$ of $S^\perp$,
\begin{equation}\label{eq:defUi}
\mmU_i= \supp(\omega^i),\qquad \mmU_0={\mmU}\setminus \bigcup\limits_{i=1}^d \mmU_i  =\{U_i\in \mmU \mid \omega_i=0, \ \textrm{for all }\omega\in S_\mmU^\perp\}.
\end{equation}
 Note that $\mmU_0$ is independent of the choice of basis of $S_\mmU^\perp$.
Consider $\Lambda_{\mmU_0}\subseteq \Lambda_{\mmU}$, that is, the set of reactions 
for which the reactant has one species in $\mmU_0$, and the product has at least two species in $\mmU_0$ or one self-interacting species in $\mmU_0$. 
In view of Lemma \ref{lemma:cons_laws},  given $i>0$, the reactant of a reaction has a species in $\mmU_i$ if and only if the product does. 
Hence the reactions in $\Lambda_{\mmU_0}$ have no species in $\mmU_i$ for any $i>0$, as any such reaction would have  two species in $\mmU$ in the reactant, one in $\mmU_i$ and one in $\mmU_0$.

We consider the multidigraphs $\mmG_{\mmU_i}$, defined from the reactant-non\-inte\-racting sets $\mmU_i$. These  should not be confused with the submultidigraphs $\mmGu|_{\mmU_i}$. The two multidigraphs might or might not agree  (for an elaboration on this, see the proof of Lemma \ref{lemma:sepcomp}).
Additionally, define the multidigraphs
\begin{equation*} 
\mmG_0=\mmG_{\mmU}|_{\mmU_0\cup\{*\}},\qquad
 \mmG_{0,i}=\mmG_\mmU|_{\mmU_0\cup\mmU_i\cup\{*\}}, \quad i=1,\dots,d.
\end{equation*}

A basis of $S_\mmU^\perp$  is said to be \emph{nonnegative} if the components of the basis vectors are nonnegative, and to have \emph{disjoint} supports if $\mmU_i\cap\mmU_j=\emptyset$ for all $i\not=j$ with $i,j>0$. 
In the latter case, the sets $\mmU_i$, $i=0,\dots,d$, form a partition of $\mmU$. Further, $S_\mmU^\perp$ is the direct sum of one-dimensional vector subspaces. It follows that the partition defined by any basis with disjoint support is independent of the basis.

The node set of a connected component of $\mmG_\mmU$ that does not contain $*$ agrees with one of the sets $\mmU_i$ for some $i$  (see  below). With this in mind, we let $C_\mmU\subseteq \{1,\dots,d\}$ be the set of indices $i>0$ such that $\mmU_i$ is the node set of some connected component of $\mmG_\mmU$ (this excludes the component with $*$).

\begin{theorem}[\textbf{Solution to the elimination system}]\label{thm:elim}
Let $\mmU$ be a reactant-noninteracting set, $\k$ a $\mmU$-linear kinetics defined on $\Rnn^m\times \Omega$, and $(u_0,x)\in\Rnn^m\times \Omega$.   
Assume that $S_\mmU^\perp$ has a nonnegative basis  $\{\omega^1,\dots,\omega^d\}$ with disjoint supports, and  let  the corresponding vector of total amounts be $T_{u_0} = (T_1,\dots,T_d)$. 
 For $i=1,\dots,d$ choose,   $j_1,\dots,j_d$ such that $U_{j_i}\in\mmU_i$, and define 
\[D(i)=\sum_{k \mid U_k\in\mmU_i}\omega^i_k\Upsilon^*_{\mmG_{\mmU_i}}\hspace{-0.05cm}(U_{j_i},U_k).\]
Then the following holds.
\begin{enumerate}[(i)] 
\item \textbf{Uniqueness criterion}. $\det(A(x)) =(-1)^{m-d}\Upsilon_{\mmG_0}(*) D(1) \cdot\ldots \cdot D(d)$.

\item \textbf{Expression}. If $\det(A(x))\not=0$, the  solution to the elimination system \eqref{eq:linelimsyst} is
\begin{align*}
u_\ell & = \dfrac{T_{k} \Upsilon^*_{\mmG_{\mmU_k}}\hspace{-0.05cm}(U_{j_k},U_\ell)}{D(k)} & \qquad \text{if } U_\ell \in \mmU_k, \ k>0, \\
u_\ell & = \frac{\Upsilon_{\mmG_0}\hspace{-0.05cm}(U_\ell)}{\Upsilon_{\mmG_0}\hspace{-0.05cm}(*)}+ {\mathlarger \sum\limits_{k=1,k\notin C_\mmU}^{d}}\frac{T_{k}
\Upsilon^*_{\mmG_{0,k}}\hspace{-0.05cm}(U_{j_k},U_\ell )}{\Upsilon_{\mmG_0}\hspace{-0.05cm}(*)D(k)} & \qquad \text{ if } U_\ell \in \mmU_0.
\end{align*}

\item \textbf{Nonnegativity}. Assume  that for every reaction $r\in\Lambda_{\mmU_0}$, at most one species $U_i\in \mathcal{U}_0$ in the product of  $r$ ultimately produces the only species in $\mmU_0$ in the reactant of $r$ via $\mathcal{U}_0$,  and, if such a species $U_i$ exists, then $(y'_r)_i=1$. 
Then $\Upsilon_{\mmG_0}\hspace{-0.05cm}(*)$, $\Upsilon_{\mmG_0}\hspace{-0.05cm}(U_\ell)$, $\Upsilon^*_{\mmG_{\mmU_k}}\hspace{-0.05cm}(U_{j_k},U_\ell)$,  $\Upsilon^*_{\mmG_{0,k}}\hspace{-0.05cm}(U_{j_k},U_\ell )$  are nonnegative  for $k=1,\dots,d$ and $\ell$ appropriately chosen. 
In particular, if $\det(A(x))\neq 0$, then the solution to the elimination system \eqref{eq:linelimsyst} is nonnegative.
 
\end{enumerate}
\end{theorem}

The proof of Theorem \ref{thm:elim} is given in Section \ref{sec:proofs}. Regarding the assumption of Theorem \ref{thm:elim}(iii),  a species $U_i$ ultimately produces another species $U_j$ via $\mmU_0$ if and only if the multidigraph $\mmG_0$ contains a path from $U_i$ to $U_j$ that does not go through $*$.
Under this assumption, $\det(A(x))\neq 0$ if and only if there exists a spanning tree rooted at $*$ in $\mmG_0$ and, for each $j\in \{1,\dots,d\}$, there exists a spanning forest of $\mmG_{\mmU_j}$ composed of two rooted trees, one with root $*$, and the other containing $U_{j_i}$. Similar conclusions can be drawn about the positivity of the solution. 

Theorem \ref{thm:elim}(ii) shows that the solution to the elimination system \eqref{eq:linelimsyst} is a rational function in $v_r(x)$ and the total amounts. If the kinetics is mass-action, then the  solution is a rational function in $x$, the reaction rate constants and the total amounts. 

If we aim to parameterize the steady state manifold by elimination of reactant-noninteracting species, then the set of species cannot contain the support of a vector in $S^\perp$, and hence  $S_\mmU=\{0\}$ and $\mmU=\mmU_0$. In this case  we obtain the simple expression
\[ u_\ell  = \frac{\Upsilon_{\mmGu}(U_\ell)}{\Upsilon_{\mmGu}(*)}, \qquad \textrm{for all }\quad U_\ell \in \mmU.\]

The theorem is illustrated on the running example, before elaborating on the necessity of a basis of $S_\mmU^\perp$ with disjoint support and discussing how to check the conditions of the theorem graphically.

\begin{myexampleReactnet}\label{example:reactnet2}
The vectors $\omega^1,\omega^2$ are nonnegative and have disjoint support. This gives $\mmU_1=\{U_1,U_2\}$, $\mmU_2=\{U_3,U_4\}$ and $\mmU_0=\{U_5\}$. Only the set $\mmU_1$ is the node set of a connected component of $\mmG_\mmU$ and hence $C_\mmU=\{1\}$. We choose $j_1=2, j_2=4$.
We consider the following multidigraphs

\medskip
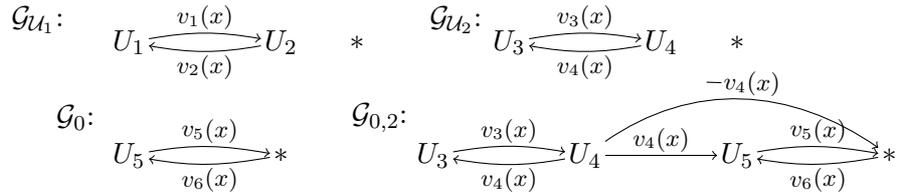
\begin{center}
\begin{tikzpicture}[inner sep=1.2pt]
\node (G1) at (-5.2,2.3) {$\mmG_{\mmU_1}$:};
\node (U2) at (-2,2) {$U_2$};
\node (U1) at (-4,2) {$U_1$};
\node (*1) at (-1,2) {$*$};
\node (G2) at (0.3,2.3) {$\mmG_{\mmU_2}$:};
\node (U3) at (1,2) {$U_3$};
\node (U4) at (3,2) {$U_4$};
\node (*2) at (4,2) {$*$};
\node (G0) at (-4.7,1) {$\mmG_0$:};
\node (U5) at (-4,0.5) {$U_5$};
\node (*) at (-2,0.5) {$*$};

\node (G02) at (-0.7,1) {$\mmG_{0,2}$:};
\node (U32) at (0,0.5) {$U_3$};
\node (U42) at (2,0.5) {$U_4$};
\node (U52) at (4,0.5) {$U_5$};
\node (*2) at (6,0.5) {$*$};

\draw[->] (U3) to[out=10,in=170] node[above,sloped]{\footnotesize $v_3(x)$}(U4);
\draw[->] (U4) to[out=190,in=-10] node[below,sloped]{\footnotesize $v_4(x)$}(U3);
\draw[->] (U5) to[out=10,in=170] node[above,sloped]{\footnotesize $v_5(x)$}(*);
\draw[->] (*) to[out=190,in=-10] node[below,sloped]{\footnotesize $v_6(x)$}(U5);
\draw[->] (U1) to[out=10,in=170] node[above,sloped]{\footnotesize $v_1(x)$}(U2);
\draw[->] (U2) to[out=190,in=-10] node[below,sloped]{\footnotesize $v_2(x)$}(U1);

\draw[->] (U32) to[out=10,in=170] node[above,sloped]{\footnotesize $v_3(x)$}(U42);
\draw[->] (U42) to[out=190,in=-10] node[below,sloped]{\footnotesize $v_4(x)$}(U32);
\draw[->] (U42) to[out=0,in=180] node[above,sloped]{\footnotesize $v_4(x)$}(U52);
\draw[->] (U42) to[out=35,in=145] node[above,sloped]{\footnotesize $-v_4(x)$}(*2);
\draw[->] (U52) to[out=10,in=170] node[above,sloped]{\footnotesize $v_5(x)$}(*2);
\draw[->] (*2) to[out=190,in=-10] node[below,sloped]{\footnotesize $v_6(x)$}(U52);
\end{tikzpicture}
\end{center}
 Then,
\[\Upsilon_{\mmG_0}(*)= v_5(x), \quad D(1)= v_1(x)+v_2(x)\quad \text{and} \quad D(2)= v_3(x)+v_4(x).\]
If $x\in \Omega\cap \R^2_{>0}$, these terms are all positive, and by Theorem \ref{thm:elim}(i) system \eqref{eq:elimsyst_example} has a unique solution. 
The solution  for $\ell=1,\dots,4$ is found using the first expression in Theorem \ref{thm:elim}(ii), considering the multidigraph $\mmG_{\mmU_1}$ for $u_1,u_2$ and the multidigraph $\mmG_{\mmU_2}$ for $u_3,u_4$.
Letting $T=T_{u_0}$, we find 
\begin{align*}
&u_1={\frac {T_1v_2(x)}{v_1(x)+v_2(x)}},
\hspace{0.6em} u_2=\frac {T_1v_1(x)}{v_1(x)+v_2(x)}, 
\hspace{0.6em} u_3={\frac {T_2v_4(x)}{v_3(x)+v_4(x)}},
\hspace{0.6em} u_4={\frac {T_2v_3(x)}{v_3(x)+v_4(x)}}. 
\end{align*}
To find $u_5$, we use the second formula in Theorem \ref{thm:elim}(ii). The multidigraph $\mmG_{0,2}$ admits only one spanning forest with a connected component rooted at $U_5$ and containing $U_4$, and the other component rooted at $*$. Namely, one connected component is $*$ and the other is identified by the only path from $U_3$ to $U_5$. Therefore
$\Upsilon^*_{\mmG_{0,2}}(U_4,U_5) = v_3(x)v_4(x)$ and $\Upsilon_{\mmG_0}(U_5)= v_6(x),$ which gives
\[u_5=\frac {T_2v_3(x)v_4(x)}{v_5(x)(v_3(x)+v_4(x))}  + \frac{v_6(x)}{v_5(x)}. \]
This solution is nonnegative. Since  $\Lambda_{\mmU_0}=\emptyset$, we could as well have reached this conclusion by employing Theorem \ref{thm:elim}(iii).

\end{myexampleReactnet}

\medskip
\noindent
 \textbf{\textit{On connected components and elements of $S_\mmU^\perp$}.}
Let $\mmH$ be the node set of a connected component of   $\mmGu$ that does not contain $*$. Then, by the definition of $\mmG_\mmU$, every edge of $\mmGu|_{\mmH}$ belongs to the first subset of $\mmE^+_\mmU$ in Definition \ref{def:graphU}. Therefore any reaction  in $\mmR_\mmH$ has exactly one species  in the reactant as well as in the product in $\mmH$  and further $\mmH$ is  noninteracting. In this situation   the results of \cite{Fel_elim} apply (in the terminology of \cite{Fel_elim}, $\mmH$ is a cut), and  there exists a nonnegative vector $\omega^\mmH \in S_\mmH^\perp\subseteq S_\mmU^\perp$ with $\supp(\omega^\mmH)=\mmH$ and   $\omega^\mmH_i = 1$ for $U_{i} \in \mmH$. Further, 
the projection of any other vector of $S_\mmU^\perp$ on the components given by $\mmH$ is a scalar multiple of $\omega^\mmH$.
For example, the multidigraph $\mmGu$ in Example \ref{example:reactnet0}  readily shows that $(1,1,0,0,0)\in S_\mmU^\perp$.

Let $\mmH_1,\dots,\mmH_k$ be the node sets of the connected components of $\mmGu$ that do not contain $*$ and $\mmH_0\cup \{*\}$ be the node set  of the connected component that contains $*$. It follows that there is a direct sum decomposition
\[S_\mmU^\perp  = \bigoplus_{i=0}^k S_{\mmH_i}^\perp,\]
with $S_{\mmH_i}^\perp$ one-dimensional for $i=1,\dots,k$.
Therefore, whether or not $S_\mmU^\perp$ admits a nonnegative basis (with disjoint support) depends on whether $S_{\mmH_0}^\perp$ does.

\subsection{On bases with nondisjoint support}

The condition on disjoint support for  the basis of $S_\mmU^\perp$ in Theorem \ref{thm:elim} can be relaxed in the following sense. 
 A nonnegative basis $\{\omega^1,\ldots,\omega^d\}$ of $S_\mmU^\perp$ is said to be \emph{minimal} if there is not another nonnegative basis $\{\widehat\omega^1,\ldots,\widehat\omega^d\}$ with strictly smaller supports, that is, $\supp(\widehat\omega^i)\subseteq\supp(\omega^{k_i})$ for some $1\le k_i\le d$, $i=1,\ldots,d$, and strict inclusion in at least one case.

\begin{theorem}\label{thm:elim2} Let $\mmN=(\mmC,\mmR)$ be a reaction network on $\mmS$.
Let $\mmU$ be a reactant-noninteracting set, $\k$ a $\mmU$-linear kinetics defined on $\R^m_{\geq 0}\times \Omega$, and $(u_0,x)\in \Rnn^m\times\Omega$.   
Assume that $S_\mmU^\perp$ has a minimal nonnegative basis  $\{\omega^1,\dots,\omega^d\}$ and $\det(A(x))\neq 0$. Let $\mmU'\subseteq \mmU$ be the subset of species that are in the support of at least two basis vectors. 
Consider the associated elimination system \eqref{eq:linelimsyst}.

\begin{enumerate}[(i)] 
\item For any $U_\ell\in\mmU'$, the solution to the elimination system \eqref{eq:linelimsyst} is $u_\ell=0$.

\item For any 
 $U_\ell \in \mmU\setminus \mmU'$, the solution $u_\ell$ to the system  \eqref{eq:linelimsyst}  
can be found using  
Theorem \ref{thm:elim}(ii) on the submultidigraph of $\mmG_{\mmU}$ induced by $\mmU''=(\mmU\setminus \mmU') \cup \{*\}$ and the induced partition $\mmU_i'' = \mmU'' \cap \mmU_i$, for $i=0,\dots,d$.

\item In particular, if the condition in Theorem \ref{thm:elim}(iii) for $\mmU$ holds, then the solution to system \eqref{eq:linelimsyst} is nonnegative. 
\end{enumerate}
\end{theorem}

The proof of Theorem~\ref{thm:elim2} is given in Section~\ref{sec:proofs}.

\begin{example}\label{ex:intersect}
Consider the following reaction network with mass-action kinetics,
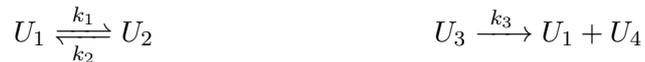
\begin{align*}
U_1\ce{<=>[k_1][k_2]}& U_2&  U_3\ce{->[k_3]}&U_1+U_4 
\end{align*}
and the reactant-noninteracting set $\mmU=\{U_1,\dots,U_4\}$. We have $S_\mmU^\perp=\langle (1,1,1,0),$ $(0,0,1,1)\rangle$, so $S_\mmU^\perp$ does not admit a  nonnegative basis with  disjoint supports. In the notation of Theorem \ref{thm:elim2}, $\mmU'=\{U_3\}$ and system \eqref{eq:linelimsyst} fulfills $u_3=0$. In order to find $u_1,u_2,u_4$ we consider the multidigraph $\mmGu$, and the submultidigraph induced by $\mmU''=\{U_1,U_2,U_4,*\}$ is obtained by removing the dashed edges and $U_3$:

\medskip
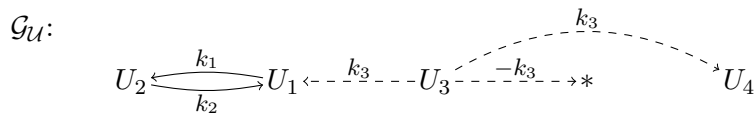
\begin{center}
\begin{tikzpicture}[inner sep=1.2pt]
\node (Gu) at (-4.3,0.7) {$\mmGu$:};
\node (v1) at (-1,0) {$U_1$};
\node (v2) at (-3,0) {$U_2$};
\node (v3) at (1,0) {$U_3$};

\node (v4) at (5,0) {$U_4$};
 
\node (*) at (3,0) {$*$};

\draw[->] (v1) to[out=170,in=10] node[above,sloped]{\footnotesize $k_1$}(v2);
\draw[->] (v2) to[out=-10,in=190] node[below, sloped]{\footnotesize $k_2$}(v1);

\draw[->,dashed] (v3) to[out=180,in=0] node[above,sloped]{\footnotesize $k_3$}(v1);
\draw[->,dashed] (v3) to[out=30,in=150] node[above, sloped]{\footnotesize $k_3$}(v4);
\draw[->,dashed] (v3) to[out=0,in=180] node[above, sloped]{\footnotesize $-k_3$}(*);

\end{tikzpicture}
\end{center}
By  Theorem \ref{thm:elim}(ii) applied to the multidigraph $\mmGu|_{\mmU'}$ and the partition $\mmU_0''=\emptyset$, $\mmU''_1 = \{U_1,U_2\}$ and $\mmU_2''=\{U_4\}$, we have
\begin{equation*}\label{eq:exzero} 
u_1=\frac {T_1k_2}{k_1+k_2},\qquad u_2=\frac {T_1k_1}{k_1+k_2},\qquad u_4=T_2.
\end{equation*}
\end{example}

\subsection{Graphical tools for reactant-noninteracting sets and Theorem \ref{thm:elim}}\label{ssec:graphtool}

In this section we define two graphs that are useful in the application of Theorem \ref{thm:elim}. In particular, 
Theorem \ref{thm:positiveRG} provides a fast way to determine whether the solution to the elimination system is nonnegative.

The first graph we introduce is to select sets of reactant-noninteracting of species. Define the \emph{interaction graph} of a reaction network   on $\mathcal{S}$ to be the labeled undirected graph with node set $\mmS$, and such that  there is an edge connecting $S_i$ and $S_j$ for $i\neq j$ if they interact, and a self-edge  for $S_i$ if it is self-interacting. The edge is solid if the end points (self-)interact in at least one reactant and dotted otherwise.
Then a set $\mmU$ is reactant-noninteracting if and only if the subgraph induced by $\mmU$ has no solid edges, and it is  noninteracting if it has no edges at all. From this construction we easily see that if $\mathcal{U}_1$ and $\mathcal{U}_2$ are reactant-noninteracting (noninteracting) sets, then $\mathcal{U}_1\cup \mathcal{U}_2$ is not necessarily reactant-noninteracting (noninteracting). Hence, there might be several reactant-noninteracting  (noninteracting) sets that are maximal in the sense that they cannot be extended further.

We introduce now a second graph, to inspect the condition for nonnegativity in Theorem \ref{thm:elim}(iii). 
We consider the \emph{reaction-coefficient multidigraph}   to be the labeled multidigraph with node set  $\mmU_0$ and edge set given by 
\[
U_i\ce{->[(r,\ (y'_r)_j)]}U_j\qquad  \text{if}\quad U_i\in\supp(y_r),\,\,\, U_j \in\supp(y'_r), \,\,\, U_i,U_j\in\mmU_0,\,\,\,r\in\mmR.   \]
Two edges corresponding to the same reaction have the same source node.  Noting that two species are in a cycle in the graph if and only if they ultimately produce each other via $\mmU_0$, then Theorem \ref{thm:elim}(iii) is equivalent to the following result (by Theorem \ref{thm:elim2}, this statement does not require the basis of $S_\mmU^\perp$ to have disjoint support).

\begin{theorem}\label{thm:positiveRG}
Let $\mmU$ be a reactant-noninteracting set, $\k$ a $\mmU$-linear kinetics defined on $\R^m_{\geq 0}\times \Omega$, and $(u_0,x)\in\Rnn^m\times\Omega$.
Assume that $S_\mmU^\perp$ has a nonnegative basis.
Furthermore, assume that for each $r\in\Lambda_{\mmU_0}$, at most one edge corresponding to $r$ is in a cycle in the reaction-coefficient multidigraph,  and, if such an edge exists, then the coefficient in its label is one. 
Then the solution to the elimination system \eqref{eq:linelimsyst} is nonnegative.
\end{theorem}

Examples of these two types of graphs are given in the next section.

\section{Examples}\label{sec:examples}
In some cases, it is possible to  reduce the task of finding  the positive solutions to the steady state equations in each stoichiometric compatibility class by  linear elimination, to the task of solving a polynomial in one variable, whose coefficients depend on $\k$ and $T$, and checking for positivity of the solutions.
This situation occurs for example in the hybrid histidine kinase model studied in \cite{feliu:unlimited}, and the allosteric kinase model analyzed in \cite{FSWFS16}. In these examples, the polynomial is found by first eliminating the concentrations of a set of reactant-noninteracting species with $n-1$ species for which Theorem \ref{thm:elim}(iii) holds. 

Linear elimination can be used to obtain nonnegative/positive parameterizations of the set of steady states. In this scenario, we aim at eliminating a   set of reactant-noninteracting species $\mmU$ of cardinality the dimension of $S$ and such that $S_\mmU^\perp=\{0\}$.  
To illustrate this, we consider a simplified phosphorelay model of the sporulation network given in \cite[Supplementary Information]{Igoshin-PNAS}. Specifically, the network consists of three main proteins KinA, Spo0F and Rap. KinA and Spo0F exist in phosphorylated and unphosphorylated form, the former indicated by the subindex $p$. Rap is a phosphatase for the dephosphorylation of Spo0F$_p$. The phosphorylated form of KinA transfers the phosphate group to Spo0F.   The reaction network consists of the following reactions with mass-action kinetics:
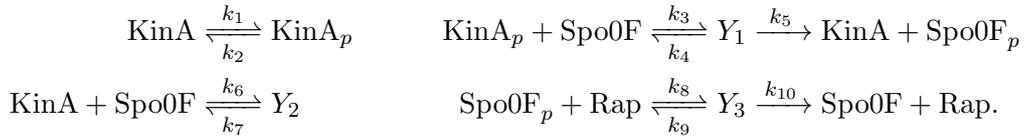
\begin{align*}
\textrm{KinA} & \ce{<=>[k_1][k_2]}   \textrm{KinA}_p  & 
\textrm{KinA}_p  + \textrm{Spo0F} & \ce{<=>[k_3][k_4]}  Y_{1 }\ce{->[k_5]}   \textrm{KinA}  + \textrm{Spo0F}_p  \\
 \textrm{KinA}  + \textrm{Spo0F}     & \ce{<=>[k_6][k_7]}  Y_{2} &
 \textrm{Spo0F}_p + \textrm{Rap}    & \ce{<=>[k_8][k_{9}]}  Y_{3} \ce{->[k_{10}]} \textrm{Spo0F} + \textrm{Rap}.  
\end{align*}
Denoting KinA, KinA$_p$, Spo0F, Spo0F$_p$, Rap by $X_1,X_2,X_3,X_4,X_5$, respectively, the
network has conservation laws $x_1+x_2+y_1+y_2=T_1$, $x_3+x_4+y_1+y_3=T_2$ and $x_5+y_3=T_3$.
The \emph{interaction graph} of this network is

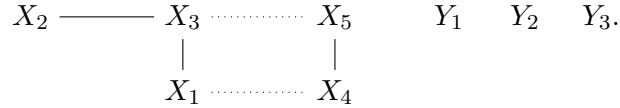
\begin{center}
\begin{tikzpicture}
\node (X2) at (0,0) {$X_2$};
\node (X3) at (2,0) {$X_3$};
\node (X5) at (4,0) {$X_5$};
\node (Y1) at (5.5,0) {$Y_1$};
\node (Y2) at (6.5,0) {$Y_2$};
\node (Y3) at (7.5,0) {$Y_3$.};
\node (X1) at (2,-1) {$X_1$};
\node (X4) at (4,-1) {$X_4$};

\draw[-] (X2) -- (X3);
\draw[-,dotted] (X1) -- (X4);
\draw[-] (X1) -- (X3);
\draw[-] (X4) -- (X5);
\draw[-,dotted] (X3) -- (X5);

\end{tikzpicture}
\end{center}
The largest sets of reactant-noninteracting species are
$\{X_1,X_2,X_4,Y_1,Y_2,Y_3\}$ and $\{X_1,X_2,X_5,Y_1,Y_2,Y_3\}$, and for both of them there exists a nonnegative basis of $S_\mmU^\perp$ with disjoint supports. The second set is noninteracting and hence fulfills Theorem \ref{thm:elim}(iii). For the first set we have $\mmU_1=\{X_1,X_2,Y_1,Y_2\}$ and $\mmU_0=\{X_4,Y_3\}$. Since $\Lambda_{\mmU_0} = \emptyset$, Theorem \ref{thm:elim}(iii) also applies. 

In order to obtain a parameterization of the steady state variety, we consider  the reactant-noninteracting set
$\mmU=\{X_1,X_2,X_4,Y_1,Y_3\}$. Now $\mmU_0=\mmU$ and $\Lambda_{\mmU_0}= \{ r_5\colon Y_1 \rightarrow X_1+X_4\}$. The \emph{reaction-coefficient multidigraph}   is

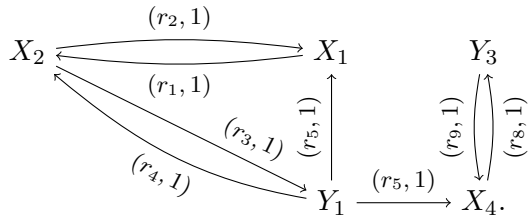
\begin{center}
\begin{tikzpicture}
\node (X2) at (-2,0) {$X_2$};
\node (X1) at (2,0) {$X_1$};
\node (Y3) at (4,0) {$Y_3$};
\node (Y1) at (2,-2) {$Y_1$};
\node (X4) at (4,-2) {$X_4$.};

\draw[->] (X2) to[out=7,in=173] node[above,sloped]{{\footnotesize $(r_2,1)$}} (X1);
\draw[->] (X1) to[out=187,in=-7] node[below,sloped]{{\footnotesize $(r_1,1)$}} (X2);

\draw[->] (X4) to[out=80,in=280] node[below,sloped]{{\footnotesize $(r_8,1)$}} (Y3);
\draw[->] (Y3) to[out=260,in=100] node[above,sloped]{{\footnotesize $(r_9,1)$}} (X4);

\draw[->] (X2) to node[above,sloped,near end]{{\footnotesize $(r_3,1)$}} (Y1);
\draw[->] (Y1) to[out=170,in=-40] node[below,sloped]{{\footnotesize $(r_4,1)$}} (X2);

\draw[->] (Y1) to node[above,sloped]{{\footnotesize $(r_5,1)$}} (X1);
\draw[->] (Y1) to node[above,sloped]{{\footnotesize $(r_5,1)$}} (X4);

\end{tikzpicture}
\end{center}
There are two edges corresponding to $r_5\in\Lambda_{\mmU_0}$.
The edge $Y_1\rightarrow X_1$ is in a cycle of the multidigraph, while the edge $Y_1\rightarrow X_4$ is not. Thus Theorem \ref{thm:positiveRG} applies and the concentration of the species in $\mmU$ can be positively expressed in terms of $x_3,x_5,y_2$. This gives a parameterization of the steady state manifold.
 In order to find the explicit expression, one can find  $\mmGu$ and apply Theorem \ref{thm:elim}(ii), or use mathematical software to solve the linear system (which  is often the fastest and most convenient option).
This analysis carries over  the complete network given in \cite{Igoshin-PNAS}, which includes phosphotransfer reactions to Spo0B and Spo0A. 

  \smallskip
We conclude with another example for which  Theorem \ref{thm:elim}(iii) fails, and in fact, the solution to the elimination system \eqref{eq:linelimsyst} is not necessarily  nonnegative. We consider the reaction network studied in \cite{kremling} for the KdpD/KdpE two-component  system in \emph{Escherichia coli} and assume mass-action kinetics:
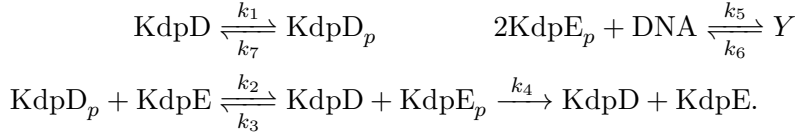
\begin{align*}
\textrm{KdpD} & \ce{<=>[k_1][k_7]} \textrm{KdpD}_p  \qquad \qquad  2\textrm{KdpE}_p + \textrm{DNA}   \ce{<=>[k_5][k_6]} Y \\  \textrm{KdpD}_p + \textrm{KdpE} & \ce{<=>[k_2][k_3]} \textrm{KdpD}+ \textrm{KdpE}_p  \ce{->[k_4]}  \textrm{KdpD} + \textrm{KdpE}.
\end{align*}
Here KdpD and KdpE are the two components of the system, which either are phosphorylated or not.  Denote KdpD,  KdpD$_p$, KdpE,  KdpE$_p$, DNA and $Y$ by $X_1,\dots,X_6$, respectively.
The \emph{interaction graph}  is
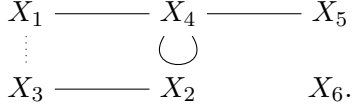
\begin{center}
\begin{tikzpicture}
\node (X1) at (0,0) {$X_1$};
\node (X4) at (2,0) {$X_4$};
\node (X5) at (4,0) {$X_5$};
\node (X3) at (0,-1) {$X_3$};
\node (X2) at (2,-1) {$X_2$};
\node (Y) at (4,-1) {$X_6$.};

\draw[-] (X4) .. controls (2.5,-0.8) and (1.5,-0.8).. (X4);
\draw[-] (X1) -- (X4);
\draw[-,dotted] (X1) -- (X3);
\draw[-] (X2) -- (X3);
\draw[-] (X4) -- (X5);

\end{tikzpicture}
\end{center}
The set $\mmU=\{X_1,X_3,X_5,X_6\}$ is reactant-noninteracting with $\mmU_1=\{X_5,X_6\}$ and $\mmU_0=\{X_1,X_3\}$. Therefore, $\Lambda_{\mmU_0}=\{r_4\colon X_1+X_4\to X_1+X_3\}$. 
The graph $\mmGu$ has two  connected components, one has node set $\mmU_1$, while the other is

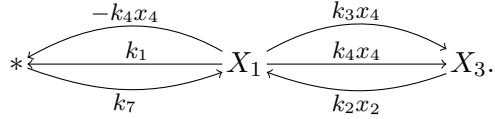
\begin{center}
\begin{tikzpicture}[inner sep=1.2pt]

\node (X1) at (0,0) {$X_1$};
\node (X3) at (3,0) {$X_3$.};
\node (*) at (-3,0) {$*$};

\draw[->] (X1) to[out=30,in=150] node[above,sloped]{\footnotesize $k_3x_4$}(X3);
\draw[->] (X1) to node[above,sloped]{\footnotesize $k_4x_4$}(X3);
\draw[->] (X3) to[out=200,in=-20] node[below,sloped]{\footnotesize $k_2x_2$}(X1);

\draw[->] (X1) to node[above]{\footnotesize \quad $k_1$}(*); 
\draw[->] (X1) to[out=150,in=30] node[above]{ \footnotesize $-k_4x_4$}(*); 
\draw[->] (*) to[out=-20,in=200] node[below]{\footnotesize $k_7$ }(X1); 
\end{tikzpicture}
\end{center}
The species $X_1\in\mmU$ is in the reactant of $r_4$.
 Both $X_1$ and $X_3$ ultimately produce $X_1$ via $\mmU_0$. Indeed, the reaction $r_4$ itself implies that $X_1$ ultimately produces $X_1$, and the reaction $r_3$ gives that $X_1$ ultimately produces $X_3$. So  Theorem \ref{thm:elim}(iii) cannot be used to conclude nonnegativity of the solution. 
In fact, using Theorem \ref{thm:elim}(ii) with $C_\mmU=\emptyset$ gives
\[x_1 = \frac{\Upsilon_{\mmG_0}(X_1)}{\Upsilon_{\mmG_0}(*)} = \frac{k_2k_7x_2 }{k_2x_2(k_1-k_4x_4)} = \frac{k_7 }{k_1-k_4x_4}.\]
If $x_4>k_1/k_4$, then this solution is negative.

%%%%%%%%%%%%%%%%%%%%%%%%
%%%   PROOFS  %%%%%%%%%%%
%%%%%%%%%%%%%%%%%%%%%%%%

\section{Proof of  Theorem \ref{thm:elim} and Theorem \ref{thm:elim2}}\label{sec:proofs}
In this section we give proofs of the main results.  Throughout this section, we assume a reaction network $(\mmC,\mmR)$ on $\mmS$ is given with a reactant-noninteracting set of species $\mmU$ and a $\mmU$-linear kinetics $\k$. In addition, we assume  a nonnegative basis  $\{\omega^1,\dots,\omega^d\}$ of $S_\mmU^\perp$  with disjoint support is given.
The proof of Theorem \ref{thm:elim} builds on the general results given in \cite{Saez:PosSol}, where the solution to a specific, though general, type of linear square systems is analyzed. 

We start by making some  simplifications on the notation to ease the readability of the proofs. 
Firstly, we do not write explicitly the dependence of $A, b, v_r$ on $x, T_{u_0}$. Secondly, we identify $U_1,\dots,U_m$ with their indices, that is,  $\mmU$ is $\{1,\dots,m\}$ and  $\mmG_\mmU$ has node set $\{1,\dots,m,*\}$.
Furthermore, we denote the cardinalities of the sets $\mmU_i$ by $m_i$,
and assume $\mmU$ is ordered such that for $ i=1,\dots,d$,
\[ \mmU_i = \left\{ 1+ \sum_{j=1}^{i-1} m_j, \dots, \sum_{j=1}^{i} m_j\right\} \quad \textrm{and}\  \quad \mmU_0= \left\{ 1+  \sum_{j=1}^{d} m_j,\dots,m\right\}. \]

The solution to the elimination system \eqref{eq:linelimsyst} is independent of the order of the equations, as well as the choice of redundant equations $\dot{u}_j=0$ to remove, for each vector $\omega^i$. Therefore, we build the system such that the equation
\[ \rho(\omega^i) \cdot u - T_i=0\] 
replaces the equation $\dot{u}_{j_i}=0$ in the system $\dot{u}_1=0,\dots,\dot{u}_m=0$ (as in Example \ref{example:reactnet}).
Then, with the chosen order,  system \eqref{eq:linelimsyst} fulfills 
\begin{equation*} 
A=\left(
\begin{array}{ccccc}
A_1 & 0 &\cdots& 0 & 0 \\
0  & A_2  & \cdots & 0 & 0  \\
\vdots & \vdots & \ddots & \vdots & \vdots \\
0  &0& \cdots & A_d & 0 \\ \hline
\multicolumn{5}{c}{A_0}  \end{array}
\right)\in \R^{m\times m}, \qquad 
b= \left(\begin{array}{c}
b^1    \\
b^2\\ 
\vdots\\
 b^d  \\ \hline
b^0  \end{array}\right)\in \R^m,
\end{equation*}
with $A_0\in \R^{m_0\times m}$ and $b^0\in \R^{m_0}$, and for  $i=1,\dots,d$,
\begin{enumerate}[(i)]
\item $A_i\in \R^{m_i\times m_i}$.
\item  $b^i \in \R^{m_i}$ has  at most one nonzero entry  (and exactly one if $T_i\neq 0$).
\end{enumerate}
It follows that this system is of the type studied  in \cite[Section 4]{Saez:PosSol}, where the sets $\mmU_i$, $i=1,\dots,d$, are denoted $\mmN_i$, and $\mmN_0$ agrees with $\mmU_0\cup \{*\}$.

Some extra definitions are required to prove Theorem \ref{thm:elim}. We let $s,t \colon \mmEu \rightarrow  \{1,\dots,m,*\}$  denote the functions assigning each edge of $\mmG_\mmU$ to its source and target, respectively. 
For a multidigraph $\mmG=(\mmN,\mmE)$ and two sets $F,B\subseteq \mmN$ with the same cardinality $M$,  let $\Theta_\mmG(F,B)$ be the set of spanning forests of $\mmG$ such that each forest has $M$ connected components (trees), each tree is rooted at a node in $B$ and contains one node in $F$. 
For finite disjoint sets 
$W_1,\dots,W_k$, define the following set of subsets
\[
W_1\odot \cdots \odot W_k=\!\underset{i\in\{1,\dots,k\}}{\mathlarger{\mathlarger\odot}}W_i =
\Big\{\!\{w_{i_1},\dots,w_{i_k}\}\subseteq \bigcup\limits_{i=1}^kW_i\!\st\! w_{i_j}\in W_j\text{ for }j=1,\dots,k\Big\},
\]
that is, all sets in $W_1\odot \cdots \odot W_k$ have $k$ elements.
Let  $\mmU_{d+1}=\{m+1\}$ and define
\begin{align*}
F & =\{j_1,\dots,j_d,m+1\}, \\
\V  & = \mmU_ 1 \odot \cdots \odot  \mmU_{d+1}, \\
\V^i & =\mmU_1  \odot \cdots \odot  \mmU_{i-1} \odot \mmU_{i+1}  \odot \cdots \odot \mmU_{d+1}, \qquad i=1,\ldots,d+1.
\end{align*}
If $B \in \V$ or $B \in \V^k$ for some $k$, then the set $B\cap \mmU_i$, ($i\neq k$ in the latter case) consists of a single element denoted by $B_{(i)}$,
that is, \[ B\cap \mmU_{i}=\{B_{(i)} \}.\]

\medskip
\paragraph{Proof of Theorem \ref{thm:elim}(i) and (ii)} 
Denote by $A|b$  the matrix obtained by appending the column $b$ to $A$.
To prove Theorem \ref{thm:elim}(i)-(ii), we use   \cite[Proposition 3]{Saez:PosSol}, which says the following. If $\mmG_\mmU$ fulfills 
\begin{itemize}
\item[(A1)] for every $i>0$, any edge with target in $\mmU_i$ has source also  in $\mmU_i$,
\item[(A2)] the $\ell$-th row of the Laplacian of $\mmG_\mmU$ agrees with the $\ell$-th row of $A|b$ for all $\ell \notin \{j_1,\dots,j_d,m+1\}$,
\end{itemize}
then the solution to \eqref{eq:linelimsyst} is
\begin{align}\label{eq:solution2}
u_\ell &=\frac{\sum\limits_{k=1}^{d+1} (-b_{j_k}) \sum\limits_{B\in \V^{k},\ell\notin B}\left(\prod\limits_{ i=1, i\neq k}^d a_{j_i B_{(i)}}\right) \Upsilon_{\mmGu}(F,B\cup \{\ell\})   }{\sum\limits_{B\in \V}\left(\prod\limits_{i=1}^d a_{j_i B_{(i)}}\right)\Upsilon_{\mmGu}(F,B)},
\end{align}
with $-b_{j_{d+1}}=-1$ \cite{Saez:PosSol}.
Further, the denominator of this expression is $(-1)^{m-d}\det(A)$. 

Therefore, the strategy to prove  Theorem \ref{thm:elim}(i)-(ii) is first to show that (A1) and (A2) hold, and then use the specific structure of $\mmG_\mmU$ to simplify the terms in \eqref{eq:solution2} to get the expressions of  Theorem \ref{thm:elim}. Note that (A1) and (A2) imply, in the terminology of \cite{Saez:PosSol}, that the multidigraph $\mmG_\mmU$ is $A$-compatible, which is a requirement to apply \cite[Proposition 3]{Saez:PosSol}.

To show that (A2) holds, let  $\widetilde{A}x+ \widetilde{b}=0$ be  as in  \eqref{eq:sseqU}.
Let $\mmE_{ji}$ be the set of parallel edges with source $j$ and target $i$ and let $\pi$ denote the labeling function of $\mmGu$.  The  Laplacian of $\mmG_\mmU$ is by definition the $(m+1)\times (m+1)$ matrix $L=(L_{ij})$ given entry-wise as follows. 
For $i,j<m+1$ and $i\neq j$, we have
\begin{align*}
L_{ij} & = \sum\limits_{e\in \mmE_{ji}}\pi(e)  = \sum_{r\in \mmR_\mmU, (y_r)_j=1}
(y_r')_i v_r=  \sum_{r\in \mmR, (y_r)_j=1} v_r (y_r'-y_r)_i = \widetilde{a}_{ij},
\end{align*}
where we have used that  if $(y_r)_j=1$, then $(y_r)_i=0$, because $\mmU$  is reactant-noninteracting. 
Next we find, for $i,j<m+1$,
\begin{align*}
L_{i,m+1} &= \sum\limits_{e\in \mmE_{m+1,i}}\pi(e)  = \sum_{r\in \mmR_\mmU, \rho(y_r)=0}
(y_r')_i v_r = \widetilde{b}_{i}, \\
L_{m+1,j} &= \sum\limits_{e\in \mmE_{j,m+1}}\pi(e)  = \sum_{r\in \mmR_\mmU, (y_r)_j=1, \rho(y'_r)=0}
 v_r +   \sum_{r\in \Lambda_\mmU, (y_r)_j=1} -\lambda_r v_r
\end{align*}
(where $\lambda_r$ is given in Definition \ref{def:graphU}).
Finally, we study  $L_{jj}$ for $j\leq m$, defined as $-\sum_{j\neq i}L_{ij}$. 
If $r\in \mmR_\mmU\setminus \Lambda_\mmU$ fulfills $ (y_r)_j=1$, then for at most one index $i\neq j$ we have $(y_r')_i=1$, and if so $(y_r')_i v_r=(y_r)_j v_r$.   Furthermore,  $(y_r'-y_r)_j\neq 0$ if and only if either $r$ belongs to $\Lambda_\mmU$ or $(y_r')_j=0$.
This gives:
\begin{align*}
L_{jj} &=  - \sum_{\substack{r\in \mmR_\mmU\\ (y_r)_j=1}}\sum\limits_{i=1, i\neq j}^m  (y_r')_i v_r    
-  \sum_{\substack{r\in \mmR_\mmU\\ (y_r)_j=1,\rho(y'_r)=0}}
 v_r +   \sum_{\substack{r\in \Lambda_\mmU\\ (y_r)_j=1}} \left( -v_r + \sum_{i=1}^m   (y_r')_i v_r \right) \\ 
& = -  \sum_{\substack{r\in \mmR_\mmU\setminus \Lambda_\mmU,  (y_r')_j=0 \\  (y_r)_j=1,\rho(y_r')\neq 0}} (y_r)_j v_r  
-  \sum_{\substack{r\in \mmR_\mmU\\  (y_r)_j=1, \rho(y'_r)=0}}
  (y_r)_j v_r +    \sum_{\substack{r\in \Lambda_\mmU \\  (y_r)_j=1}} - (y_r)_j  v_r +  (y_r')_j v_r  \\
&= \sum_{r\in \mmR_\mmU, (y_r)_j=1} (y_r'-y_r)_j v_r = \widetilde{a}_{jj}.
\end{align*}
 
Consequently, the first $m$ rows of $L$ agree with the matrix $\widetilde{A}|\widetilde{b}$, which in turn agrees with the matrix $A|b$, except for the rows $j_1,\dots,j_d$. Thus (A2) holds.

 In order to prove (A1) and later Theorem \ref{thm:elim}(iii) below, we state a general lemma.

\begin{lemma}\label{lemma:facts} The multidigraph $\mmGu$ fulfills:
\begin{enumerate}[(i)]
\item \label{F1} Every edge with target in  $\mmU_i$ for $i>0$ has source also in $\mmU_i$.
\item \label{F2} Every edge from $\mmU_i$ for $i\neq 0$ to $\mmU_0\cup \{*\}$  corresponds to a reaction in $\Lambda_\mmU$ whose product has at least one species in $\mmU_i$.
\item \label{F3} Let $\zeta$ be a spanning forest of $\mmGu$ and $\tau$ a connected component of $\zeta$. If $\tau$ contains a node in $\mmU_i$, $i\neq 0$, then its root is either in $\mmU_i$ or in $\mmU_0\cup \{*\}$. If $\tau$ contains a node in $\mmU_0\cup \{*\}$, then its root is also in $\mmU_0\cup \{*\}$.
\end{enumerate}
\end{lemma}
\begin{proof}
Let $e\in\mmEu^+$ with target in $\mmU_i$ and $r\in \mmR_\mmU$ the associated reaction. By Lemma \ref{lemma:cons_laws} with $\mmH=\mmU_i$, $s(e)\in\mmU_i$. Since the edges in $\mmEu^-$ have target node $*$, statement (\ref{F1}) and (\ref{F3}) follow. 
For (ii), if $s(e)\in\mmU_i$, then by Lemma \ref{lemma:cons_laws} the associated reaction $r$ has at least one product in $\mmU_i$.This guarantees that $r$ defines another edge with target in  $\mmU_i$. Hence $r$ must be in $\Lambda_\mmU$ and statement (\ref{F2})  follows.
\end{proof}

By Lemma \ref{lemma:facts}(i), (A1) holds.   Thus the solution to the elimination system \eqref{eq:linelimsyst} is as given in \eqref{eq:solution2}. 
The next lemma is useful for  simplifying \eqref{eq:solution2}.

\begin{lemma}\label{lemma:sepcomp}
It holds that  
\begin{gather} 
\Upsilon_{\mmGu}(F,B)=\Upsilon_{\mmG_0}(*)\prod_{i=1}^d \Upsilon^*_{\mmG_{\mmU_i}}(j_i, B_{(i)}),  \qquad B\in \V, \label{eq:sepcompA} \\
\Upsilon_{\mmGu}(F,B \cup \{\ell\})= \Upsilon_{\mmG_0}(\ell)\prod_{i=1}^d\Upsilon^*_{\mmG_{\mmU_i}}(j_i, B_{(i)}), \qquad  B\in \V^{d+1},\ \ell>j_d,\label{eq:sepcompB} \\
\Upsilon_{\mmGu}(F,B\cup \{\ell\})=\Upsilon^*_{\mmG_{0,k}}(j_k,\ell)\prod^d_{\begin{subarray}{c}i=1\\ i\neq k\end{subarray}}\Upsilon^*_{\mmG_{\mmU_i}}(j_i, B_{(i)}), \quad  B\in \V^k,\  k\leq d, \ j_d<\ell\leq m.\label{eq:sepcompC}
\end{gather}
\end{lemma}

\begin{proof}
For all three equalities, the term on the left-hand side depends on $\mmGu$, while the terms on the right-hand side depend on the multidigraphs $\mmG_{\mmU_i}$, $i>0$, as well as the submultidigraphs $\mmG_0$ and $\mmG_{0,k}$ of $\mmGu$. Since $\mmG_{\mmU_i}$ is not necessarily a submultidigraph of $\mmGu$, we start by comparing them. Consider   $\mmG_\mmU$,  and $\mmG_{\mmU_i}$, for $i>0$, with node sets  $\mmU\cup \{*\}$ and $\mmU_i\cup \{*\}$, respectively.  There is a natural label-preserving correspondence between the set of edges between nodes in $\mmU_i$ of the two multidigraphs. By Lemma \ref{lemma:facts}(\ref{F2}), any edge from a node in $\mmU_i$ to the node $*$ in  $\mmGu$ or $\mmG_{\mmU_i}$ 
corresponds to a reaction in $\Lambda_\mmU$, whose product has a species in $\mmU_i$.  Therefore, in neither multidigraph there are edges from $\mmU_i$ to $*$ with positive label. 

For $r\in \Lambda_\mmU$ such that the reactant has a species in $\mmU_i$, let $e$ be the corresponding edge of $\mmG_\mmU$ from a node in $\mmU_i$ to $*$. Then $e$ is also an edge of  $\mmG_{\mmU_i}$ if and only if  $r\in\Lambda_{\mmU_i}$, that is, the product of $r$ has at least two species in $\mmU_i$, or a self-interacting species in $\mmU_i$. 
In this case, the label of the edge  is $-\lambda_r v_r(x)$ in $\mmGu$ and 
 $-\lambda^i_r v_r(x)$ in $\mmG_{\mmU_i}$ with 
\[ \lambda_r^i=  \sum_{j \st U_j \in \mmU_i} (y_r')_j -1.\]
This is the only difference between the multidigraph $\mmG_{\mmU_i}$ and the submultidigraph $\mmG_\mmU|_{\mmU_i\cup \{*\}}$ of $\mmG_\mmU$ induced by $\mmU_i\cup \{*\}$. 

Consider now the following cases: 
\begin{align*}
\text{(P1)}\qquad  &  \widetilde{B} =B,  && B \in \V, \\
\text{(P2)}\qquad  &  \widetilde{B}= B\cup \{\ell\},&&  B\in \V^{d+1}, \ \ell>j_d,\\
\text{(P3)}\qquad  &  \widetilde{B}= B\cup \{\ell\}, && B\in \V^k,\ k\leq d, \ m\geq \ell>j_d.
\end{align*}
The terms on the left-hand side of \eqref{eq:sepcompA}-\eqref{eq:sepcompC} arise from the labels of the spanning forests in  $  \Theta_{\mmG_\mmU}(F,\widetilde{B})$. 
Let $D=\{1,\dots,d\}$ in  the case (P1) and (P2), and $D=\{1,\dots,d\}\setminus\{k\}$ in  the case (P3). Let $\mmG_\mmU^D$ be the multidigraph obtained from $\mmGu$ as follows. For each edge from a node in $\mmU_i$, $i\in D$, to the node $*$ in $\mmGu$, if $\lambda^i_r=0$, then the edge is removed, and if not, the label of this edge in $\mmGu^D$ is defined as $-\lambda^i_r v_r(x)$. Furthermore, remove from $\mmGu$ all edges from a node in $\mmU_i$,  $i\in D$, to a node in $\mmU_0$. The constraints on $\mmGu$ in Lemma \ref{lemma:facts} imply 
 that the expressions on the right-hand side of  \eqref{eq:sepcompA}-\eqref{eq:sepcompC} agree with
$ \Upsilon_{\mmG_\mmU^D}(F,\widetilde{B}).$ 
Therefore we need to show that 
\begin{equation}\label{eq:upsilonD}
\Upsilon_{\mmG_\mmU}(F,\widetilde{B}) =  \Upsilon_{\mmG_\mmU^D}(F,\widetilde{B}).
\end{equation}

Let $\zeta'\in  \Theta_{\mmG_\mmU}(F,\widetilde{B})$. Consider an edge $e\colon j\rightarrow j'$ with $j\in \mmU_i$ for some $i\in D$ and $j'\in \mmU_0$. By Lemma \ref{lemma:facts}(\ref{F2}), the reaction $r$ corresponding to this edge  belongs to $\Lambda_\mmU$ and hence, this reaction gives rise to an additional edge $j\rightarrow *$  in $\mmE_\mmU^-$. Replacing the edge $j \rightarrow j'$ of  $\zeta'$ by the edge $j\rightarrow *$  gives a new element $\zeta\in \Theta_{\mmG_\mmU}(F,\widetilde{B})$, since no cycle is created according to Lemma \ref{lemma:facts}(\ref{F1}) and there are no edges from $\mmU_0$ to $\mmU\setminus \mmU_0$.

Let $\Gamma_{\mmG_\mmU}(F,\widetilde{B})$ be the set of spanning forests $\zeta$
that do not have any edge with source in $\mmU_i$ for some $i\in D$ and target in $\mmU_0$. Consider the map $\gamma \colon \Theta_{\mmG_\mmU}(F,\widetilde{B}) \rightarrow \Gamma_{\mmG_\mmU}(F,\widetilde{B})$ 
that maps a spanning forest $\zeta'$ to the spanning forest obtained by replacing all edges with source in $\mmU_i$, $i\in D$, and target in $\mmU_0$, by the corresponding edges with target $*$ (in $\mmE_\mmU^-$), as explained above. In the   case (P3), no edge from a node in $\mmU_k$ to a node in $\mmU_0$ is  changed since $k\notin D$.
This map is surjective and gives the decomposition
\[  \Theta_{\mmG_\mmU}(F,\widetilde{B}) = \bigsqcup\limits_{\zeta\in \Gamma_{\mmG_\mmU}(F,\widetilde{B})} \gamma^{-1}(\zeta).\]
Therefore
\begin{align}\label{eq:upsiliondecomp}
\Upsilon_{\mmG_\mmU}(F,\widetilde{B}) &= \sum_{\zeta\in \Gamma_{\mmG_\mmU}(F,\widetilde{B}) }\ \sum_{\zeta'\in \gamma^{-1}(\zeta)} \pi(\zeta').
\end{align} 

Let  $\zeta\in \Gamma_{\mmG_\mmU}(F,\widetilde{B})$, and let $\alpha(\zeta)$ be the set of edges of $\zeta$ with source in $\mmU_i$, $i\in D$, and target $*$. For $e\in\alpha(\zeta)$, let $\beta_e$ be the union of $\{e\}$ with the set of edges in $\mmGu$ with source $s(e)$ and target in  $\mmU_0$, corresponding to the same reaction as $e$.   Then, there is a one-to-one correspondence 
 \[ \gamma^{-1}(\zeta)  \leftrightarrow \underset{e\in  \alpha(\zeta)}{\mathlarger{\mathlarger\odot}} \beta_e,\]
  since every spanning forest in $\gamma^{-1}(\zeta)$ is obtained by replacing edges of $\alpha(\zeta)$ with edges corresponding to the same reactions but with targets in $\mmU_0$.
Hence
\begin{equation}\label{eq:zetaprime2}
\sum_{\zeta'\in \gamma^{-1}(\zeta)} \pi(\zeta')   =   \left(   \prod_{e\in \zeta\setminus \alpha(\zeta) } \pi(e)  \right)
\left(   \prod_{e\in  \alpha(\zeta) }  \sum_{e'\in \beta_e} \pi(e') \right).  
\end{equation}
Furthermore, if $e$ corresponds to the reaction $r$ and $s(e)\in \mmU_i$, then
\begin{equation}\label{eq:zetaprime}
\sum_{e'\in \beta_e} \pi(e') = \pi(e) + \sum_{e'\in \beta_e | t(e')\in \mmU_0}\pi(e')= -\lambda_r v_r(x) + \sum_{ j\in \mmU_0} (y_r')_jv_r(x)  
= -\lambda_r^i v_r(x),  
\end{equation}
where it is used that $(y_r')_j\not=0$  only if $j\in \mmU_i\cup \mmU_0$.
Since $\zeta$ does not have an edge from $\mmU_i$ to $\mmU_0$ for any $i\in D$ by definition, all labeled edges in $\zeta\setminus \alpha(\zeta)$ also belong  to $\mmGu^D$. We have further just shown that for $e\in \alpha(\zeta)$ such that   \eqref{eq:zetaprime} is nonzero, there is an edge in  $\mmGu^D$ with label \eqref{eq:zetaprime}. 
This implies that if \eqref{eq:zetaprime2} is different from zero, then the spanning forest $\zeta$ is naturally identified with a spanning forest in $\mmGu^D$ in $\Theta_{\mmG}(F,\widetilde{B})$. Hence, using \eqref{eq:upsiliondecomp}, the equality \eqref{eq:upsilonD} holds 
and the proof is  completed.
\end{proof}

We can now prove Theorem \ref{thm:elim}(i)-(ii)   using \eqref{eq:solution2} and the previous lemma.
Using the definition of $A$ and $b$, we have $a_{j_iB_{(i)} }=\omega^i_{B_{(i)}}$, and $-b_{j_k}=T_k$ for $k=1,\dots,d$. 
By \eqref{eq:sepcompA}, the denominator of \eqref{eq:solution2} is 
\begin{multline} {\sum\limits_{B\in \V}\left(\prod\limits_{i=1}^d a_{j_i B_{(i)}}\right)\Upsilon_{\mmGu}(F,B)}
= \Upsilon_{\mmG_0}(*) \sum\limits_{B\in \V}\left(\prod\limits_{i=1}^d \omega^i_{B_{(i)}}  \Upsilon^*_{\mmG_{\mmU_i}}(j_i, B_{(i)}) \right)  \\
=  \Upsilon_{\mmG_0}(*) \prod_{i=1}^d \sum_{k\in \mmU_i} \omega^i_{k}  \Upsilon^*_{\mmG_{\mmU_i}}(j_i, k)  =  \Upsilon_{\mmG_0}(*) D(1) \cdot \dots \cdot D(d). \label{eq:detA}
\end{multline}
Since the denominator of \eqref{eq:solution2} is equal to $(-1)^{m-d}\det(A)$, we obtain Theorem \ref{thm:elim}(i).

Consider  now the numerator of \eqref{eq:solution2} for a fixed $\ell$, and assume $\ell \in \mmU_k$, $k>0$.
We easily see that  $\Theta_{\mmGu}(F,B\cup \{\ell\})=\emptyset$ if  $B\in\V^i$ with 
$i\in \{1,\dots,d+1\}\setminus\{k\}$, by using Lemma~\ref{lemma:facts}(iii)  and that $B$ has two elements  in $\mmU_k$ while $F$ only one. 
Hence  $\Upsilon_{\mmGu}(F,B\cup \{\ell\})=0$ for $B\in \V^i$ and $i\in \{1,\dots,d+1\}\setminus\{k\}$.
Now, if $B\in \V^k$, $B\cup \{\ell\}$ belongs to $\V$ with $B_{(k)}=\ell$, and we use \eqref{eq:sepcompA}  to rewrite $\Upsilon_{\mmGu}(F,B\cup \{\ell\})$. 
The numerator of \eqref{eq:solution2} becomes
\begin{multline*} 
  \sum\limits_{B\in \V^{k}}\left(\prod\limits_{ \substack{ i=1 \\ i\neq k} }^d \omega^i_{B_{(i)}}  \right) \Upsilon_{\mmGu}(F,B\cup \{\ell\})    
 = \Upsilon_{\mmG_0}(*) \Upsilon^*_{\mmG_{\mmU_k}}(j_k,\ell)  \sum\limits_{B\in \V^{k}} \prod\limits_{  \substack{ i=1 \\ i\neq k} }^d \omega^i_{B_{(i)}} \Upsilon^*_{\mmG_{\mmU_i}}(j_i, B_{(i)})    \\
 = \Upsilon_{\mmG_0}(*) \Upsilon^*_{\mmG_{\mmU_k}}(j_k,\ell)   \prod\limits_{ \substack{ i=1 \\ i\neq k} }^d \sum\limits_{j\in \mmU_i }    \omega^i_{j} \Upsilon^*_{\mmG_{\mmU_i}}(j_i, j) =  \Upsilon_{\mmG_0}(*) \Upsilon^*_{\mmG_{\mmU_k}}(j_k,\ell)  \prod\limits_{  \substack{ i=1 \\ i\neq k}}^d  D(i).  
  \end{multline*}
Combining this with \eqref{eq:detA}, we obtain that, for $\ell \in \mmU_k$ with $k>0$, 
\[ u_\ell = \frac{T_k \Upsilon^*_{\mmG_{\mmU_k}}(j_k,\ell)  }{D(k)},\]
as desired. Finally, consider  the case $\ell>j_d$. 
Using \eqref{eq:sepcompB}-\eqref{eq:sepcompC},  the term
\[ \sum\limits_{B\in \V^{k},\ell\notin B}\left(\prod\limits_{ i=1, i\neq k}^d \omega^i_{B_{(i)}}\right) \Upsilon_{\mmGu}(F,B\cup \{\ell\})\]
 agrees with
\[\begin{cases}
\Upsilon^*_{\mmG_{0,k}}(j_k,\ell) 
  \sum\limits_{B\in \V^{k}}\left(\prod\limits_{\substack{i=1 \\  i\neq k } }^d \omega^i_{B_{(i)}} \Upsilon^*_{\mmG_{\mmU_i}}(j_i, B_{(i)}) \right) 
  =   \Upsilon^*_{\mmG_{0,k}}(j_k,\ell)    \prod\limits_{ \substack{ i=1 \\ i\neq k}}^d  D(i), & k\neq d+1, \\
   \Upsilon_{\mmG_0}(\ell) 
  \sum\limits_{B\in \V^{k}}\left(\prod\limits_{ \substack{ i=1 \\ i\neq k}}^d \omega^i_{B_{(i)}} \Upsilon^*_{\mmG_{\mmU_i}}(j_i, B_{(i)}) \right) 
  =   \Upsilon_{\mmG_0}(\ell)   \prod\limits_{ \substack{ i=1}}^d  D(i), & k=d+1.
  \end{cases}\]
  If $k\in C_\mmU$, then there is not an edge from  $\mmU_k$ to $\mmU_0$, and hence   $\Upsilon^*_{\mmG_{0,k}}(j_k,\ell) =0$, since a rooted tree cannot contain both $j_k$ and $\ell$. 
  Using \eqref{eq:detA},  \eqref{eq:solution2} becomes, for $\ell>j_d$,
 \[ u_\ell = \frac{ \Upsilon_{\mmG_0}(\ell) }{\Upsilon_{\mmG_0}(*)} +   {\mathlarger \sum\limits_{k=1, i\notin C_\mmU}^{d}}\frac{T_{k}
\Upsilon^*_{\mmG_{0,k}}(j_k,\ell )}{\Upsilon_{\mmG_0}(*)D(k)}. \]
  This concludes the proof of Theorem \ref{thm:elim}(ii). 
 \begin{flushright}
$\square$
\end{flushright}

\medskip
\paragraph{Proof of Theorem \ref{thm:elim}(iii)} 
We proceed to prove the last statement of the theorem, using the ideas in the proof of the first part and again using 
\cite{Saez:PosSol}. According to \eqref{eq:solution2}, the solution is nonnegative if all terms  $\Upsilon_{\mmG_\mmU}(F,B\cup \{\ell\})$ and $\Upsilon_{\mmG_\mmU}(F,B)$ in the expression are nonnegative, since the $j_i$-th row of $A$ ($\omega^i$) and the entries of $b$ are nonnegative. In  \cite{Saez:PosSol}, conditions on   $\mmG_\mmU$ are given that guarantee this.
Therefore, the strategy is to show that the condition in Theorem \ref{thm:elim}(iii) implies the conditions for nonnegativity in \cite{Saez:PosSol}. 

The multidigraph $\mmG_\mmU$ fulfills (i) the set of edges is a disjoint union of the set of positive and negative edges, $\mmE_\mmU= \mmEu^+ \sqcup \mmEu^-$; (ii) all cycles in $\mmGu$ contain at most one edge in $\mmEu^-$, since negative edges always have target $*$; and  (iii) any path in $\mmG_\mmU$ that contains a negative edge, contains $*$. 

For $\ell\in \mmU$, define
\[ \mathcal{V}_\ell=\{j\in\mmU\st j\ \text{does not ultimately produce}\ \ell\ \text{via}\ \mmU\},\]
and consider the map $\mu\colon\mmEu^- \rightarrow \mathcal{P}(\mmEu^+)$   defined by 
\begin{equation*} 
 \mu\Big(i\ce{->[-\lambda_rv_r(x)]}*\Big)= \Big\{ i\ce{->[(y'_r)_jv_r(x)]}j\in\mmEu^+ \st  j\in \mathcal{V}_i \Big\}.
\end{equation*}
This map fulfills 
\begin{enumerate}[(a)]
\item\label{cond3aPgraph} if $e'\in\mu(e)$, then $s(e)=s(e')$,
\item\label{cond3bPgraph}  if $e'\in\mu(e)$, then every cycle containing $e'$ contains $t(e)$ (since there is not a path from $t(e')$ to $s(e)$ that does not go through $*$, by definition of $\mu$). 
\item\label{cond3cPgraph} if $e\neq e'$ then $\mu(e)\cap\mu(e')=\emptyset$. 
\end{enumerate}

 In the terminology of \cite{Saez:PosSol}, (i), (ii) and (a)-(c) imply that the pair $(\mmGu,\mu)$ is an \emph{edge partition}.  If  
 \medskip
 \begin{enumerate}[(iv)]
\item \label{cond4Pgraph} $\pi(e)+\sum\limits_{e'\in\mu(e)}\pi(e')\in\Rnn$  for all $e\in\mmEu^-$
\end{enumerate} 
also holds, then $\mmGu$ is called \emph{P-graph} with associated map $\mu$. In that case, since $\mmGu$ fulfills (A1) and (A2) (is $A$-compatible, as argued in the previous proof) and fulfills condition (iii), the assumptions of \cite[Theorem 5]{Saez:PosSol} are fulfilled, guaranteeing nonnegativity of the solution.  It further transpires from the last equation of the proof of \cite[Theorem 5]{Saez:PosSol} that  (\hyperref[cond4Pgraph]{iv}) in fact implies   $\Upsilon_{\mmG_\mmU}(F,\widetilde{B}) \in \Rnn$ for all $\widetilde{B}$ in expression  \eqref{eq:solution2}. Therefore, the goal is to show that the condition of Theorem \ref{thm:elim}(iii) implies (\hyperref[cond4Pgraph]{iv}), in which case we are done.

Let $e\in \mmEu^-$ with $s(e)=\ell$, and let $r$ be the associated reaction. 
We have
\begin{equation}\label{eq:almostPgraph}
\pi(e)+\sum\limits_{e'\in\mu(e)}\pi(e') = \left(-\lambda_r +\sum_{j\in \mathcal{V}_\ell}(y'_r)_{j}\right) v_r(x)=\left(1-\sum_{j\in \mmU\setminus\mathcal{V}_\ell}(y'_r)_{j}\right) v_r(x).
\end{equation}
By \eqref{eq:almostPgraph},  condition (\hyperref[cond4Pgraph]{iv}) holds for $e$ if  $\sum_{j\in \mmU\setminus\mathcal{V}_\ell}(y'_r)_{j}\leq 1$, or, equivalently, if

\medskip
 \begin{itemize}
\item[($\triangle$)] \label{condedges} there is at most one species $j\in \mmU$ in the product of $r$ that ultimately produces $s(e)$, and further, if such a species exists, $(y_r')_j=1$.
\end{itemize} 

\medskip
The reactions in $\Lambda_{\mmU_0}$ correspond to the edges $e\in \mmEu^-$  with  $s(e)\in \mmU_0$. Thus, the assumption in Theorem \ref{thm:elim}(iii) is precisely that (\hyperref[condedges]{$\triangle$}) holds for all $e\in \mmEu^-$ such that $s(e)\in \mmU_0$.
Hence, we  need  to show  that (\hyperref[condedges]{$\triangle$}) always holds for all $e\in \mmEu^-$ with $s(e)\in \mmU\setminus \mmU_0$. 
For this, let $e\in \mmEu^-$ with $s(e)=\ell \in \mmU_i$, $i>0$, and let $r$ be the associated reaction.   
By Lemma \ref{lemma:facts}(i), any path that ends in the node $\ell$ has all nodes in $\mmU_i$.
So, to show that (\hyperref[condedges]{$\triangle$}) holds we only need to focus on the species $j$ in the product of $r$ that belong to $\mmU_i$. 

Since $\omega^i$ is nonnegative with support $\mmU_i$, we have
\[
0=\omega^i\cdot(y'_r-y_r)=\omega^i_\ell((y_r')_\ell-1)+\sum_{k\in \mmU_i\setminus\{\ell\}}(y_r')_k\,\omega^i_k.
\]
Thus, either 
\[ (y_r')_\ell=1,\quad \text{and}\quad \sum\limits_{k\in \mmU_i\setminus\{\ell\}}(y_r')_k\,\omega^i_k=0,\]
 or 
\[ (y_r')_\ell=0,\quad \text{and}\quad\sum\limits_{k\in \mmU_i\setminus\{\ell\}}(y_r')_k\,\omega^i_k=\omega^i_\ell.\]
In the first case, $(y_r')_k=0$ for all $k\in \mmU_i\setminus\{\ell\}$ since $\omega^i_k> 0$, and (\hyperref[condedges]{$\triangle$}) holds since $(y_r')_\ell=1$.
In the second case, $\ell$ is not in the product of $r$. Assume there exist $\ell_1,\dots,\ell_j\in \mmU_i$, $j\geq 1$, in the
 product of $r$ such that $\ell_q$, $q=1,\ldots,j$, ultimately produces $\ell$ via $\mmU$ (hence via $\mmU_i$) for all $q=1,\dots,j$.
For each $q$, the sequence of reactions associated with the (simple) path from $\ell_q$ to $\ell$ fulfills that exactly one of the reactions involves $\ell_q$ in the reactant, none of the reactants involve $\ell$ and at least one product involves $\ell$,  and any other species in $\mmU_i$ that is in the reactant of one of the reactions (thus with stoichiometric coefficient one), is also in the product of some reaction. 
This implies that  the sum of the reaction vectors of the  sequence of reactions defines a vector $z^q\in S\cap \mathbb{Z}^n$ such that 
$z^q_{\ell_q}\geq -1$, $z^q_\ell\geq 1$, and $z^q_k\geq 0$ for $k\in \mmU_i\setminus\{\ell,\ell_q\}$. 
Now, using that $(y_r'-y_r) \cdot \omega^i =0$ and $z^q\cdot \omega^i =0$ for  $q=1,\dots,j$, it holds  that
\begin{align*}
0=(y_r'-y_r) \cdot \omega^i  =& -\omega^i_\ell+ \sum_{q=1}^j (y_r')_{\ell_q}\,\omega^i_{\ell_q}  +\sum_{k\in \mmU_i\setminus\{\ell,\ell_1,\dots,\ell_q\}}(y_r')_k\,\omega^i_k, \\
0=z^q \cdot \omega^i=&\,\,z^q_\ell \omega^i_\ell+z_{\ell_q}^q\,\omega^i_{\ell_q}+\sum_{
k\in \mmU_i\setminus\{\ell,\ell_q\}}z^q_k\,\omega^i_k, \qquad q=1,\dots,j.
\end{align*}
Since $(y_r')_{\ell_q}> 0$, the sum of the right-hand sides of these $j+1$ equalities is necessarily strictly positive (hence, the system is incompatible), unless $j=1$ and $(y_r')_{\ell_1}=1$. 
This shows that (\hyperref[condedges]{$\triangle$}) holds for the given $e$, which concludes the proof of Theorem \ref{thm:elim}(iii).
  \begin{flushright}
$\square$
\end{flushright}

\medskip
 
\paragraph{Proof of Theorem \ref{thm:elim2}}
(i) Let $\ell \in \mmU'$ and assume for simplicity that $\ell \in \mmH= \mmU_1\cap \mmU_2\subsetneq \mmU_1$. 
We show first that 
$S_{\mmH }^\perp=\{0\}$.
Indeed, if this is not the case, then there is a  nonzero vector $z\in S_{\mmH}^\perp\subset S_{\mmU}^\perp$. 
For $\lambda=\min_{i\in \mmU_1}(z_i/\omega^1_i)$, the vector $\omega'=-\lambda  \omega^1+ z$ is nonnegative and $\supp(\omega')\subsetneq \mmU_1$. Since  $\{\omega',\omega^2,\dots,\omega^d\}\subset \R^n_{\geq 0}$ is a basis of $S_\mmU^\perp$, we reach a contradiction because $\{\omega^1,\dots,\omega^d\}$ is minimal by assumption.

Consider the system \eqref{eq:linelimsyst}  for $\mmH$. 
Since $S_{\mmH}^\perp=\{0\}$,  Theorem \ref{thm:elim} applies. Lemma \ref{lemma:cons_laws} applied to $\omega^1$ and $\omega^2$, respectively, implies that for any reaction with  a species in $\mmH$ in the product, there must be one species in $\mmU_1$ and one in $\mmU_2$, respectively, in the reactant. Since $\mmU$ is reactant-noninteracting,  this species must be in $\mmH=\mmU_1\cap \mmU_2$. Therefore,   an edge of $\mmG_{\mmH}$ with target  in $\mmH$ cannot have source $*$, and in particular,  there is not a spanning forest of $\mmG_{\mmH}$ with root $\ell\in\mmH$.
 By Theorem \ref{thm:elim}(ii), we conclude that 
 $u_\ell=0$. Since the solution to system \eqref{eq:linelimsyst} for $\mmH$ agrees with  the solution to the original system for $\mmU$, we have shown (i).

(ii) We choose $j_1,\dots,j_d$ not in $\mmU'$ (this is always possible).  According to (i), the variables corresponding to the species in $\mmU'$ can be equated to zero in \eqref{eq:linelimsyst} and the resulting reduced system has full rank.
Consider  the reaction network $\mmN'$ on $\mmS\setminus \mmU'$ constructed from the original network by removing the reactions with species in $\mmU'$ in the reactant.
This automatically removes all reactions with a product involving  species in $\mmU'$ by Lemma \ref{lemma:cons_laws}, with argument as in (i).
Let $S'$ be the stoichiometric subspace of $\mmN'$.  System \eqref{eq:linelimsyst} 
associated with $\mmN'$ and the induced kinetics 
agrees with system \eqref{eq:linelimsyst} 
for the original network, after letting $u_\ell=0$ for all $\ell \in \mmU'$. Indeed, 
 the kinetics is $\mmU$-linear and thus the rate of the reactions with species in $\mmU'$ in the reactant equate to zero. 
Thus, outside the rows corresponding to $\omega^1,\dots,\omega^d$, the reduced system agrees with the elimination system for $\mmN'$ and $\mmU\setminus \mmU'$.  

Let  $p$ be the projection from $\R^n$ onto the components of the species in $\mmS\setminus \mmU'$.  The vectors  $p(\omega^1),\dots,p(\omega^d)$ are  linearly independent because they have disjoint support and further each $p(\omega^i)$ is   orthogonal to the vectors of the reactions in $\mmN'$. This defines $d$ independent vectors in $(S')_{\mmU\setminus \mmU'}^\perp$. Since the reduced system has full rank, $(S')_{\mmU\setminus \mmU'}^\perp$ has dimension $d$.
This demonstrates that the reduced system is the elimination system of $\mmN'$ associated with $\mmU\setminus \mmU'$.
The associated multidigraph $\mmG'_{\mmU\setminus \mmU'}$ agrees with the submultidigraph  of $\mmGu$
induced by $\mmU\setminus \mmU'$, because, in particular, any edge from  $\mmU\setminus \mmU'$ to $*$ in $\mmGu$ arises from a reaction not involving $\mmU'$.  Hence (ii) is proven. 

(iii) If the condition in Theorem \ref{thm:elim}(iii) holds for $\mmN$ and $\mmU$, then it also holds for $\mmN'$ and $\mmU\setminus \mmU'$ 
since $\mmU'\cap \mmU_0=\emptyset$.  This implies nonnegativity of $u_\ell$ for all $\ell \in \mmU\setminus \mmU'$, and combined with $u_\ell=0$ for $\ell\in \mmU'$, we obtain  (iii).

\small

%\bibliographystyle{plain}
%\bibliography{crnt}

\end{document}